\documentclass[a4paper,11pt]{article}

\usepackage{amsmath}
\usepackage{amsfonts}
\usepackage{amssymb}         
\usepackage{amsopn}
\usepackage{amsthm}
\usepackage{graphicx}
\usepackage{epsfig}
\usepackage{times}
\usepackage{verbatim}
\usepackage{mathtools}
\usepackage{subcaption}
\usepackage{authblk}
\usepackage{hyperref}
\usepackage[left=1.0in,top=1.2in,right=1.0in]{geometry}
\usepackage{enumerate}
\usepackage{rotating}
\usepackage{hyperref}
\usepackage{qcircuit}
\usepackage{color}
\usepackage{lineno}

\usepackage{array}
\usepackage{bbm}

\usepackage{longtable}
\newcolumntype{L}{>{$}l<{$}}

\usepackage{algorithm}
\usepackage{algpseudocode}

\newtheorem{theorem}{Theorem}[section]

\newtheorem{lemma}[theorem]{Lemma}

\newcommand{\tr}{\text{Tr}}
\newcommand{\bra}[1]{\mathinner{\langle{#1}\rvert}}
\newcommand{\ket}[1]{\mathinner{\lvert{#1}\rangle}}
\newcommand{\braket}[1]{\mathinner{\langle{#1}\rangle}}
\newcommand{\ketbra}[2]{\ket{#1}\bra{#2}}

\makeatletter
\newsavebox{\@brx}
\newcommand{\lAngle}[1][]{\savebox{\@brx}{\(\m@th{#1\langle}\)}%
    \mathopen{\copy\@brx\kern-0.5\wd\@brx\usebox{\@brx}}}
\newcommand{\rAngle}[1][]{\savebox{\@brx}{\(\m@th{#1\rangle}\)}%
    \mathclose{\copy\@brx\kern-0.5\wd\@brx\usebox{\@brx}}}
\makeatother

\makeatletter
\newsavebox{\@strut}
\newcommand{\BBra}[1]{%
    \sbox{\@strut}{\(#1\)}%
    \mathinner{\left\langle\kern-0.3\ht\@strut\left\langle{#1}\right|\right.}%
}
\newcommand{\KKet}[1]{%
    \sbox{\@strut}{\(#1\)}%
    \mathinner{\left.\left|{#1}\right\rangle\kern-0.3\ht\@strut\right\rangle}%
}
\newcommand{\BBraket}[1]{%
    \sbox{\@strut}{\(#1\)}%
    \mathinner{\left\langle\kern-0.3\ht\@strut\left\langle{#1}\right\rangle\kern-0.3\ht\@strut\right\rangle}%
}
\newcommand{\KKetbra}[2]{%
    \sbox{\@strut}{\(#1#2\)}%
    \mathinner{\left.\left|{#1}\vphantom{#2}\right\rangle\kern-0.3\ht\@strut\right\rangle%
        \left\langle\kern-0.3\ht\@strut\left\langle{#2}\vphantom{#1}\right|\right.}%
}
\makeatother

\title{An alternative approach to optimal wire cutting \\ without ancilla qubits}

\author{Edwin Pednault}
\affil{IBM T.J.~Watson Research Center, Yorktown Heights, NY}

\date{}

\begin{document}


\maketitle


\begin{abstract}
  Wire cutting is a technique for partitioning large quantum circuits
  into smaller subcircuits in such a way that observables for the
  original circuits can be estimated from measurements on the smaller
  subcircuits.  Such techniques provide workarounds for the limited
  numbers of qubits that are available on near-term quantum devices.
  Wire cutting, however, introduces multiplicative factors in the
  number of times such subcircuits need to be executed in order to
  estimate desired quantum observables to desired levels of
  statistical accuracy.  An optimal wire-cutting methodology has
  recently been reported that uses ancilla qubits to minimize the
  multiplicative factors involved as a function of the number of wire
  cuts.  Until just recently, the best-known wire-cutting technique
  that did not employ ancillas asymptotically converged to the same
  multiplicative factors, but performed significantly worse for small
  numbers of cuts.  This latter technique also requires inserting
  measurement and state-preparation subcircuits that are randomly
  sampled from Clifford 2-designs on a per-shot basis.  This paper
  presents a modified wire-cutting approach for pairs of subcircuits
  that achieves the same optimal multiplicative factors as wire
  cutting aided by ancilla qubits, but without requiring ancillas.
  The paper also shows that, while unitary 2-designs provide a
  sufficient basis for satisfying the decomposition, 2-designs are not
  mathematically necessary and alternative unitary designs can be
  constructed for the decompositions that are substantially smaller in
  size than 2-designs.  As this paper was just about to be released, a
  similar result was published, so we also include a comparison of the
  two approaches.

\end{abstract}

\section{Introduction}
Current quantum computing devices implement limited numbers of qubits.
If a quantum circuit requires more qubits than is available on a given
device, the circuit cannot be executed directly.  Various technical
approaches have been developed to address this situation, one of which
is known as wire cutting.

Wire cutting partitions a large quantum circuit into smaller
subcircuits that can then execute on small quantum devices.
Partitioning is achieved by ``cutting'' some of the connections
between quantum gates, and then replacing those connections with
measurement and state-preparation operations.  By cutting a sufficient
number of such connections, a set of disconnected subcircuits can be
created, whereby each subcircuit can then be executed on a desired
target device that implements a limited number of qubits.  After
execution, the measurements made for each subcircuit can be
post-processed using classical computation to estimate a desired
quantum observable for the overall circuit, almost as though the
entire circuit was executed on a single, sufficiently large, quantum
device without wire cutting.

However, wire cutting incurs a cost.  Each cut introduces a
multiplicative factor in the number of times the subcircuits need to
be executed in order to estimate desired quantum observables to
desired levels of statistical accuracy.  Quantum observables are
calculated classically by applying post-processing functions that map
actual measurement outcomes (i.e., binary 1's and 0's) to real numbers
and then averaging the resulting real numbers across multiple
executions.  The averaging serves to estimate expected values, and the
statistical estimation errors in those estimates increase with each
wire cut by multiplicative factors.  To compensate, the number of
circuit executions must then be increased by corresponding
multiplicative factors in order to achieve the same levels of
statistical accuracy as would be obtain without wire cutting.  A
central issue in wire-cutting research is to therefore develop
measurement and state-preparation protocols that minimize the
multiplicative factors incurred.

Wire cutting was originally proposed by Peng {\it et
  al.}~\cite{Peng2020wirecutting}.  The basic idea is to decompose the
identity channel into a linear combination of measurement and state
preparation operations.  Wire cutting is performed by inserting an
identity gate at the point where a cut is to be made, and then
replacing the identity gate with its linear decomposition.  The
decomposition of the identity channel for a single qubit may be
expressed mathematically as
\begin{equation}
  \label{eq:Peng}
  \text{id}(\rho) = \rho = \sum_{j=1}^T c_j \rho^j \tr\left(O^j \rho\right)
  \; ,
\end{equation}
where each $O^j$ is a Hermitian observable, $\rho$ and $\rho^j$ are
density matrices, and each $c_j$ is a real-valued coefficient.  The
observables $O^j$ correspond to measurements, while the density
matrices $\rho^j$ correspond to state preparations.  The coefficients
$c_j$ are multiplying weights that are used when calculating the net
overall observable for the original circuit from the measurements made
on the partitioned subcircuits.

\begin{figure}[tb]
  \centering
  \includegraphics[width=0.54\textwidth]{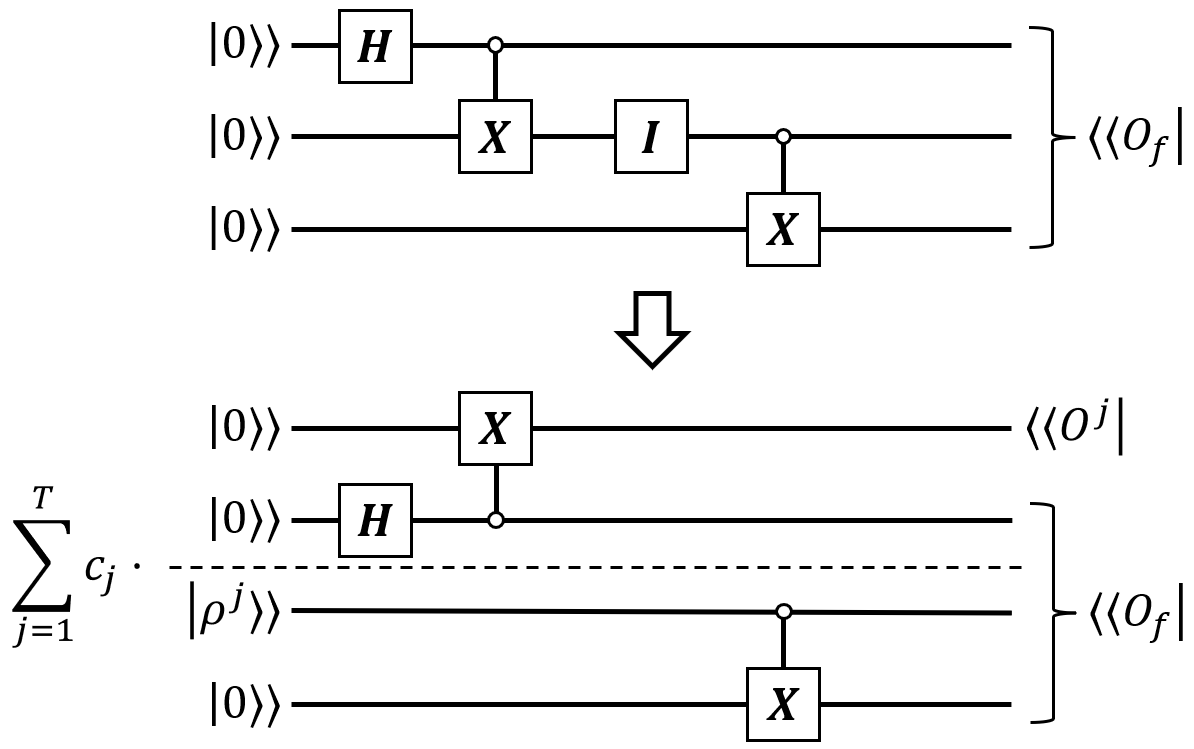}
  \caption{An example of performing wire cutting on a circuit where
    the objective is to estimate an observable $O_f$.  An identity
    gate is inserted at the point where a wire cut is to be made.  The
    gate is then replaced with its linear decomposition, which in this
    example results in two disconnected subcircuits and the
    introduction of an additional observable $O^j$ and state
    preparation $\rho^j$.  The resulting subcircuits can be executed
    separately, with their measurements combined during
    post-processing to estimate the resulting overall observable $O^j
    \otimes O_f$ for the decomposed circuit.}
  \label{fig:wire_cutting_example}
\end{figure}

Fig.~\ref{fig:wire_cutting_example} illustrates the wire cutting
process for a simple circuit where the objective is to estimate an
observable $O_f$.  In this example, wire cutting creates two
disconnected subcircuits and introduces an additional observable $O^j$
and state preparation $\rho^j$ per Eq.~\ref{eq:Peng}.  If
$\mathcal{U}$ is the channel implemented by the original circuit and
if $\mathcal{U}_A$ and $\mathcal{U}_B$ are the channels implemented by
the top and bottom disconnected subcircuits, respectively, then the
expected value of $O_f$ can be expressed in terms of the expected
values of the observables in the disconnected subcircuits as
\begin{equation}
  \langle O_f \rangle
  = \tr\left(O_f \mathcal{U}(\rho_0^{\otimes 3})\right)
  = \sum_{j=1}^T c_j \tr\left((O^j \otimes O_f)
  \left(\mathcal{U}_A(\rho_0^{\otimes 2})
  \otimes \mathcal{U}_B(\rho^j \otimes \rho_0) \right)
  \right)
  \; ,
\end{equation}
where $\rho_0 = \ketbra{0}{0}$.

A notable feature of the decomposition of Peng {\it et al.} is that
the state preparations $\rho^j$ do not depend on the measurement
outcomes of the observables $O^j$.  The decomposition therefore does
not require classical communication between subcircuits.  The same is
true of immediately subsequent
work~\cite{Ayral2021divideandcompute,Tang2021cutqc,Perlin2021maxlikelihood}
in which variations on the same decomposition were proposed.  A recent
result by Brenner {\it et al.}~\cite{Brenner2023optimalwirecutting}
shows that, for any decomposition without communication, estimation
errors must necessarily grow by multiplicative factors of at least
$(4^n)^2$ for $n$ wire cuts, where $(4^n)^2$ is the optimum
multiplicative factor achievable across all possible wire-cutting
protocols without communication for $n$ wire cuts.  The original
decomposition of Peng {\it et al.}\ and its subsequent
variations~\cite{Ayral2021divideandcompute,Tang2021cutqc,Perlin2021maxlikelihood}
all achieve this optimum.  

To do better, classical communication must be employed.  A more recent
wire-cutting decomposition by Lowe {\it et
  al.}~\cite{Lowe2023fastquantumcircuit} does consider classical
communication between subcircuits and achieves a multiplicative factor
of $(2^{n+1}+1)^2$.  Their approach combines two key insights.
\begin{itemize}
    \item State preparations on one side of a wire cut can be made
      contingent on the measurement outcomes obtained on the other
      side of the cut, which requires communication and can lead to
      more efficient decompositions.

    \item When cutting several wires that connect two subcircuits, it
      is more efficient to decompose a single identity gate for the
      entire group of wires than it is to decompose separate identity
      gates, one for each individual wire being cut.
\end{itemize}
Thus, parallel cuts across $n$ qubits are treated as a single cut
across a $d$-dimensional complex space, $d = 2^n$.  Their
decomposition employs measure-and-prepare channels of the following
form expressed in terms of unitaries $U \in \mathsf{U}(\mathbb{C}^d)$
in dimension $d$
\begin{equation}
  \label{eq:LoweMeasurePrep}
  \Psi_U(\rho) = \sum_{k=1}^d \braket{k|U^\dag \rho U |k} U \ketbra{k}{k} U^\dag
  \; .
\end{equation}
Given a probability distribution $\{(p_j,U_j)\}$ over unitaries that
forms a 2-design~\cite{Dankert2009twodesigns}, they then define a pair
of channels in dimension $d$
\begin{equation}
  \label{eq:LoweChannelMeasurePrep}
  \Psi_0(\rho) = E_U[\Psi_U]
  = \sum_j p_j \sum_{k=1}^d \braket{k|U_j^\dag \rho U_j |k} U_j \ketbra{k}{k} U_j^\dag
\end{equation}
and
\begin{equation}
  \label{eq:LoweChannelDepolarized}
  \Psi_1(\rho) = {1 \over d} \tr(\rho) \mathbbm{1}
  \; ,
\end{equation}
and they show that, if $Z$ is a Bernoulli random variable, $Z \in
\{0,1\}$ with probability of outcome $1$ being $d/(2d+1)$, then
\begin{equation}
  \label{eq:LoweIdentity}
  \text{id}(\rho) = (2d + 1) E_Z\left[ (-1)^Z \Psi_Z(\rho) \right]
  \; .
\end{equation}
When cutting $n$ qubits, $d = 2^n$ and the unitaries $U_j$ define
paired measurement and state-preparation operations on the $n$ qubits
being cut.  The multiplicative factor affecting statistical estimation
error is then $(2d + 1)^2 = (2^{n+1} + 1)^2$.

Lowe {\it et al.}\ show that Eq.~\ref{eq:LoweIdentity} holds for any
probability distribution $\{(p_j,U_j)\}$ over unitaries that forms a
2-design.  Such designs exist for all dimensions
$d$~\cite{Dankert2009twodesigns}.  In particular, the uniform
distribution over Clifford circuits forms a unitary 2-design, and it
also forms a 3-design when the dimension $d$ is a power of
two~\cite{Webb2021Clifford}.  Lowe {\it et al.}\ recommend using
uniformly random Clifford circuits as the 2-designs in their
decomposition given that efficient methods exist for generating
pseudorandom Clifford circuits~\cite{VanDenBerg2021Cliffords}.
Implementing this approach, however, would require generating and
applying a separate Clifford circuit for each individual shot, which
can be problematic for existing control circuitry.

This paper presents two improvements to the decomposition of Lowe {\it
  et al}.  The first improvement is an alternative decomposition of
the identity channel of the form
\begin{equation}
  \label{eq:OurIdentity}
  \text{id}(\rho) = (2d - 1) E_U\left[ \Phi_U(\rho) \right]
  \; .
\end{equation}
The corresponding multiplicative factor is $(2d - 1)^2 = (2^{n+1} -
1)^2$, which matches the optimal multiplicative factor achieved by
Brenner {\it et al.}~\cite{Brenner2023optimalwirecutting} using
wire-cutting decompositions that leverage ancilla qubits, but in
our case without using ancillas.  It should be noted,
however, that this improvement only pertains to wire cuts between
individual pairs of subcircuits (i.e., ``parallel'' wire cuts),
whereas the approach of Brenner {\it et al.}\ achieves this
multiplicative factor across all wire cuts between a subcircuit and
all of its neighbors in the resulting communication graph (i.e.,
``arbitrary'' wire cuts).  Nevertheless, the improvement presented
here still offers a significant reduction in multiplicative factors
versus Lowe {\it et al.}, particularly for small numbers of wire cuts
$n$.

The second improvement is that we show that the probability
distribution $\{(p_j,U_j)\}$ over unitaries does not have to be a
2-design, and that it is in fact possible to construct alternative
unitary designs $\{(p_j,U_j)\}$ that satisfy the decomposition that
contain substantially fewer unitaries.  In particular, we present an
algorithm for constructing such unitary designs that, in experiments
conducted thus far, has generated surprisingly compact designs.

This result is interesting from a theoretical perspective in that it
establishes, by way of example, that unitary 2-designs are sufficient
but not mathematically necessary in order to satisfy
Eqs.~\ref{eq:LoweIdentity} and~\ref{eq:OurIdentity}.  It is also
interesting from a practical perspective in that the result raises the
possibility of compiling compact libraries of unitary designs that can
be used when applying Eqs.~\ref{eq:LoweIdentity}
and~\ref{eq:OurIdentity} in practice, thereby avoiding the need to
generate uniformly random Clifford circuits on a per-run, per-shot basis.

As this paper was just about to be released, a very similar result was
published by Harada {\it et al.}~\cite{Harada2023optimalwirecutting}.
They employ a decomposition of the identity channel that can be
written in quasiprobability form as
\begin{equation}
  \label{eq:HaradaIdentity}
  \text{id}(\rho)
  = (2d-1) \left( \sum_{j=1}^d {1 \over 2d - 1}
          \sum_{k=1}^d \braket{k|U_j^\dag \rho U_j |k} U_j \ketbra{k}{k} U_j^\dag
  - {d-1 \over 2d-1} \sum_{k=1}^d \braket{k| \rho |k} \rho_k \right)
  \; ,
\end{equation}
where
\begin{equation}
  \label{eq:HaradaMixedState}
  \rho_k = \sum_{\ell=1}^d {1 \over d-1}(1 - \delta_{k,\ell}) \ketbra{\ell}{\ell}
  \; , \qquad
  \delta_{i,j} = \begin{cases}
    1 & i = j \\
    0 & i \not= j
    \end{cases}
  \; ,
\end{equation}
and where $\{U_j\}_{j=1}^d \cup \{I^{\otimes n}\}$ denotes a set of
unitary operators in dimension $d$ that transform the computational
bases into $d+1$ mutually unbiased bases~\cite{Gokhale2020mub}.  The
corresponding multiplication factor is thus the same as that achieved
here of $(2d - 1)^2 = (2^{n+1} - 1)^2$ for $n$ qubits.

If we define $U_0 = I^{\otimes n}$, Eq.~\ref{eq:HaradaIdentity} can be
rewritten as
\begin{align}
  \label{eq:HaradaLowe}
  \text{id}(\rho)
  &= (d + 1) \sum_{j=0}^d {1 \over d + 1} \sum_{k=1}^d \braket{k|U_j^\dag \rho U_j |k} U_j \ketbra{k}{k} U_j^\dag
  - d \left( {1 \over d} \tr(\rho) \mathbbm{1} \right) \notag \\
  &= (2d + 1) E_Z\left[ (-1)^Z \Psi_Z(\rho) \right]
  \; .
\end{align}
Thus, the decomposition of Harada {\it et al.}\ can likewise be seen
as an improvement to the decomposition of Lowe {\it et al}, where the
corresponding unitary design is $\{(p, U_j)\}_{j=0}^d$, and where $p =
1/(d+1)$, $U_0 = I^{\otimes n}$, and the unitary operators in the
design transform the computational bases into $d+1$ mutually unbiased
bases.  Per Eq.~\ref{eq:HaradaLowe}, this unitary design likewise
satisfies both the original decomposition of Lowe {\it et al.}\ and
the decomposition presented here.  Further relationships among these
decompositions are discussed at the end of this paper.

\section{Improvements to the Decomposition}
An improved quasiprobability decomposition can be obtained by
algebraically combining the depolarizing channel $\Psi_1$ in the
decomposition of Lowe {\it et al.}\ with the measure-and-prepare
channel $\Psi_0$ by replacing the pure-state preparations in $\Psi_0$
with mixed-state preparations, and by introducing a sign (i.e., $\pm
1$) term to account for the negative weighting assigned to the
depolarizing channel $\Psi_1$ in Eq.~\ref{eq:LoweIdentity}.  To define
this mixed-state preparation, we define $q_{\ell|k}$ to be the
conditional probability of preparing state $\ket{\ell}$ after
measuring outcome $\ket{k}$, where
\begin{equation}
  \label{eq:conditionaldist}
  q_{\ell|k} =
  \begin{cases}
    {d \over 2d-1}&\ell = k \\
    {1 \over 2d-1}&\ell \not= k
  \end{cases}
  \; .
\end{equation}
In addition, we define the sign multiplier $s_{\ell|k}$ to be
\begin{equation}
  \label{eq:conditionalsign}
  s_{\ell|k} =
  \begin{cases}
    1&\ell = k \\
    -1&\ell \not= k
  \end{cases}
  \; .
\end{equation}
We then define the signed mixed-state measure-and-prepare operation
$\Phi_U$ to be
\begin{equation}
  \label{eq:PhiUdefinition}
  \Phi_U(\rho) = \sum_{k=1}^d  \braket{k|U^\dag \rho U |k}
          \sum_{\ell=1}^d q_{\ell|k}s_{\ell|k} U \ketbra{\ell}{\ell} U^\dag
  \; .
\end{equation}
If it were not for the sign term $s_{\ell|k}$, $\Phi_U$ would be
completely positive and trace preserving, but because of 
$s_{\ell|k}$ it is not.  However, we can prove the following.
\begin{lemma}
  \label{lem:equality}
For any probability distribution $\{(p_j,U_j)\}$ over unitaries (not
necessarily a 2-design), and for Bernoulli random variable $Z \in
\{0,1\}$ with probability of outcome $1$ being $d/(2d+1)$, it is the
case that
\begin{equation}
  (2d - 1) E_U\left[ \Phi_U(\rho) \right]
  = (2d + 1) E_Z\left[ (-1)^Z \Psi_Z(\rho) \right]
  \; .
\end{equation}
\end{lemma}
\begin{proof}
  From the above definitions we have
\begin{align}
  (2d - 1) E_U\left[ \Phi_U(\rho) \right]
  &= \sum_j p_j \sum_{k=1}^d  \braket{k|U_j^\dag \rho U_j |k}
          \sum_{\ell=1}^d (2d-1)q_{\ell|k}s_{\ell|k} U_j \ketbra{\ell}{\ell} U_j^\dag \notag \\
  &= \sum_j p_j \sum_{k=1}^d  \braket{k|U_j^\dag \rho U_j |k}
          \left(
            (d+1) U_j \ketbra{k}{k} U_j^\dag 
            - \sum_{\ell=1}^d U_j \ketbra{\ell}{\ell} U_j^\dag
          \right) \notag \\
          &= (d+1) \sum_j p_j \Psi_{U_j}(\rho) - d  \sum_j p_j {1 \over d} \tr(\rho) \mathbbm{1} \notag \\
          &= (2d + 1) \left( {d+1 \over 2d + 1}\Psi_0 - {d \over 2d + 1}\Psi_1 \right) \notag \\
          &= (2d + 1) E_Z\left[ (-1)^Z \Psi_Z(\rho) \right]
  \; .
\end{align}
\end{proof}

\begin{lemma}
  \label{lem:identity}
  For any 2-design $\{(p_j,U_j)\}$
  \begin{equation}
    \text{\rm id}(\rho) = (2d - 1) E_U\left[ \Phi_U(\rho) \right]
  \end{equation}
\end{lemma}
\begin{proof}
  Follows immediately from the above lemma and Lemma~2.1
  in~\cite{Lowe2023fastquantumcircuit}.
\end{proof}

\begin{algorithm}
    \caption{Single Monte Carlo Trial}
    \label{alg:montecarlo}
    \hspace*{\algorithmicindent} \textbf{Input:}
    $f:\{0,1\}^n \rightarrow [-1,1]$,
    quantum circuit $C$ on $n$ qubits, unitary design $\{(p_i,U_i)\}$ \\
    \hspace*{\algorithmicindent} \textbf{Output:}
        Random Variable $Y$ s.t. $E[Y] = \tr(O_f\mathcal{N}(\rho_0^{\otimes n}))$
    \begin{algorithmic}[1] 
      \State Initialize wires $A$ to $\ket{0}_A$
      \State Apply circuit $C_A$ to wires $A$
      \State $V \gets$ \,unitary $U_i$ on $k$ qubits randomly drawn according to $\{(p_i,U_i)\}$
      \State Apply circuit $V^\dag$ to wires $A \cap B$
      \State $y \gets$ \,measurement of wires $A \cap B$
      \State $x_A \gets$ \,measurement of wires $[n] \setminus B$
      \State $z \gets \ell$ randomly drawn according to the conditional
             distribution $q_{\ell | y}$ defined in Eq.~\ref{eq:conditionaldist}
      \State Initialize wires $B$ to $\ket{z}_{A \cap B} \ket{0}_{[n] \setminus A}$
      \State Apply circuit $V$ to wires $A \cap B$
      \State Apply circuit $C_B$ to wires $B$
      \State $x_B \gets$ \,measurement of wires $B$
      \State $x \gets x_Ax_B$
      \State $Y \gets (2^{k+1} - 1) \, s_{z|y} \, f(x)$, where $s_{z|y}$
      is defined in Eq.~\ref{eq:conditionalsign}
      \State \textbf{return} $Y$
    \end{algorithmic}
\end{algorithm}

In practice, the identity decomposition in Lemma~\ref{lem:identity}
can be implemented as a Monte Carlo simulation that combines classical
computation with quantum execution.  An algorithm for implementing a
single trial of this simulation is specified in pseudocode in
Algorithm~\ref{alg:montecarlo}.  This algorithm would then be executed
multiple times and the resulting outputs $Y$ would be averaged over
the resulting trials to estimate the expected value of a desired
observable.  Algorithm~\ref{alg:montecarlo} borrows notation and
terminology from Lowe {\it et al.}\ for ease of comparison.  Thus, we
consider an $n$-qubit circuit $C$ composed of subcircuits $C_A,C_B$
acting on sets of qubits $A,B \subseteq [n]$, respectively, where $|A
\cup B| = n$ and $|A \cap B| = k$.  We also consider a diagonal
observable $O_f$ with eigenvalues defined by a function $f:\{0,1\}^n
\rightarrow [-1,1]$ such that $O_f = \sum_{x\in\{0,1\}^n} f(x)
\ketbra{x}{x}$.  The overall action of the quantum circuit $C$ is to
implement a quantum channel denoted by $\mathcal{N}$, and the goal is
therefore to estimate $\tr(O_f \mathcal{N}(\rho_0^{\otimes n}))$,
$\rho_0 = \ketbra{0}{0}$.  The pseudocode notation in each line in
Algorithm~\ref{alg:montecarlo} likewise follows Lowe {\it et al}.

Comparing Algorithm~\ref{alg:montecarlo} to Algorithm 1 of Lowe {\it
  et al.}~\cite{Lowe2023fastquantumcircuit}\footnote{It should be
  noted that, when specifying their Algorithm~1, Lowe {\it et
    al.}\ failed to include the equivalent of Step 9 in the above
  algorithm of applying circuit $V$ to wires $A \cap B$ after
  initializing the qubits for circuit $C_B$.  This step should have
  been inserted between their Steps 12 and 13 as required by the
  definition of their measure-and-prepare channel $\Psi_0$.  Strictly
  speaking, in their case, this step would also have to be applied
  only in the case when $z=0$, since $z$ is technically being used in
  their algorithm to choose between applying $\Psi_0$ versus $\Psi_1$.
  Pragmatically speaking, however, the conditioning on $z$ could be
  omitted because applying a unitary to a depolarized state yields a
  depolarized state in the case when $z=1$, so the conditioning really
  wouldn't matter.}, we can see that in both cases, a random variable
$z$ is introduced that controls a stochastic initialization of wires
$B$ prior to executing subcircuit $C_B$.  The only significant
procedural difference between the two algorithms is that this
stochastic initialization is conditioned on the measurement outcome of
subcircuit $C_A$ in the case of Algorithm~\ref{alg:montecarlo} above,
whereas it is independent of the measurement outcome in the case of
Algorithm 1 in~\cite{Lowe2023fastquantumcircuit}.  The mathematical
consequences of this procedural difference are that, in the case of
Algorithm~\ref{alg:montecarlo} above, the output $Y$ is scaled by a
factor $2^{k+1}-1$ and the sign change $s_{z|y}$ is conditioned on the
measurement outcome of subcircuit $C_A$, whereas for Algorithm 1
in~\cite{Lowe2023fastquantumcircuit} the output $Y$ is scaled by a
factor $2^{k+1}+1$ and the sign change $(-1)^z$ is independent of the
measurement outcome of subcircuit $C_A$.

The scaling factors $2^{k+1}-1$ and $2^{k+1}+1$ correspond to the
$\gamma$'s in the quasiprobability
representations~\cite{Temme2017errormitigation} of these
decompositions of the identity channel.  The importance of $\gamma$
for both algorithms is that the probability of the estimation error
of the observable exceeding $\epsilon$ after averaging over $N$
trials can be bounded by the Hoeffding inequality using $\gamma$:
\begin{equation}
  Pr\left\{\left|\tr(O_f\mathcal{N}(\rho_0^{\otimes n}))
     - {1 \over N} \sum_{i=1}^N Y_i \right| \ge \epsilon \right\}
  \le 2 e^{-{N \epsilon^2 \over 2\gamma^2}}
  \; .
\end{equation}
Although the difference between $\gamma = 2^{k+1}-1$ for
Algorithm~\ref{alg:montecarlo} above and $\gamma = 2^{k+1}+1$ for
Algorithm~1 in~\cite{Lowe2023fastquantumcircuit} becomes
asymptotically negligible as the number of wire cuts $k$ gets large,
for one wire cut it means a $\gamma$ of $3$ instead of $5$ and, hence,
a multiplicative factor in the Hoeffding exponent of $9$ versus $25$.
For two cuts it means a $\gamma$ of $7$ versus $9$ and a
multiplicative factor of $49$ versus $81$.  So, for small numbers of
wire cuts, the statistical differences between the two are quite
significant, with the decomposition presented here offering a
substantial improvement over the decomposition presented
in~\cite{Lowe2023fastquantumcircuit}.

\begin{algorithm}
    \caption{Alternative Single Monte Carlo Trial}
    \label{alg:altmontecarlo}
    \hspace*{\algorithmicindent} \textbf{Input:}
    $f:\{0,1\}^n \rightarrow [-1,1]$,
    quantum circuit $C$ on $n$ qubits, unitary design $\{(p_i,U_i)\}$ \\
    \hspace*{\algorithmicindent} \textbf{Output:}
        Random Variable $Y$ s.t. $E[Y] = \tr(O_f\mathcal{N}(\rho_0^{\otimes n}))$
    \begin{algorithmic}[1] 
      \State Initialize wires $A$ to $\ket{0}_A$
      \State Apply circuit $C_A$ to wires $A$
      \State $V \gets$ \,unitary $U_i$ on $k$ qubits randomly drawn according to $\{(p_i,U_i)\}$
      \State Apply circuit $V^\dag$ to wires $A \cap B$
      \State $y \gets$ \,measurement of wires $A \cap B$
      \State $x_A \gets$ \,measurement of wires $[n] \setminus B$
      \State $w \gets 1$ with probabilty $2^k/(2^{k+1}-1)$, $0$ otherwise
      \If{$w=0$}
          \State Initialize wires $B$ to $\ket{y}_{A \cap B} \ket{0}_{[n] \setminus A}$
          \State $z \gets y$
      \Else
          \State Initialize wires $B$ to $\ket{0}_{A \cap B} \ket{0}_{[n] \setminus A}$
          \State Apply Hadamard gates to wires $A \cap B$
          \State $z \gets$ \,measurement of wires $A \cap B$
      \EndIf
      \State Apply circuit $V$ to wires $A \cap B$
      \State Apply circuit $C_B$ to wires $B$
      \State $x_B \gets$ \,measurement of wires $B$
      \State $x \gets x_Ax_B$
      \State $Y \gets (2^{k+1} - 1) \, s_{z|y} \, f(x)$, where $s_{z|y}$
      is defined in Eq.~\ref{eq:conditionalsign}
      \State \textbf{return} $Y$
    \end{algorithmic}
\end{algorithm}

It is worth noting that the random initialization defined in Steps~7
and~8 of Algorithm\ref{alg:montecarlo} can also be implemented in a
slightly different manner by leveraging the fact that the conditional
probability distribution $q_{\ell | y}$ defined in
Eq.~\ref{eq:conditionaldist} can be expressed as a mixture of a
deterministic distribution and a uniformly random distribution
\begin{equation}
  \label{eq:conditionaldist}
  q_{\ell|k} = \left({d-1 \over 2d-1}\right) q^0_{\ell|k} +  \left({d \over 2d-1}\right) q^1_{\ell|k}
  \; , \qquad
  q^0_{\ell|k} = \begin{cases}
    1&\ell = k \\
    0&\ell \not= k
  \end{cases}
  \; , \qquad
  q^1_{\ell|k} = {1 \over d}
  \; .
\end{equation}
Thus, initialing wires $A \cap B$ in circuit $C_B$ to $\ket{\ell}$
according to the distribution $q_{\ell|k}$ is equivalent to
initializing those wires to $\ket{k}$ with probability $(d-1)/(2d-1)$
and to a uniformly random state (i.e., completely depolarized) with
probability $d/(2d-1)$.  The latter can be accomplished by
initializing wires $A \cap B$ to the ground state $\ket{0}$, applying
Hadamard gates, and then measuring those qubits.

Algorithm~\ref{alg:altmontecarlo} embodies this alternative approach,
which may be easier to implement on existing quantum devices than
Algorithm~\ref{alg:montecarlo}. With this change, $\Phi_U$ in
Eq.~\ref{eq:PhiUdefinition} is effectively being re-expressed as
\begin{align}
  \label{eq:PhiUdefinition}
  \Phi_U(\rho)
  &= {d-1 \over 2d-1}\sum_{k=1}^d  \braket{k|U^\dag \rho U |k}
           U \ketbra{k}{k} U^\dag
    + {d \over 2d-1} \sum_{k=1}^d  \braket{k|U^\dag \rho U |k}
          \sum_{\ell=1}^d s_{\ell|k} {1 \over d} U \ketbra{\ell}{\ell} U^\dag \notag \\
  &= \left({d-1 \over 2d-1}\right)\Psi_U
    + {d \over 2d-1} \sum_{k=1}^d  \braket{k|U^\dag \rho U |k}
          \sum_{\ell=1}^d s_{\ell|k} {1 \over d} U \ketbra{\ell}{\ell} U^\dag
  \; ,
\end{align}
where $\Psi_U$ is the measure-and-prepare channel of Lowe {\it et
  al.}\ defined in Eq.~\ref{eq:LoweMeasurePrep}.

\section{Improvements to the Unitary Design}
Our second result demonstrates that, while unitary 2-designs are
sufficient to satisfy the equational form of the decomposition
proposed by Lowe {\it et al.}, they are not mathematically necessary.
The mapping from unitaries $U$ to measure-and-prepare channels
$\Psi_U$ is, in fact, highly injective (i.e., many to one).  There are
two fundamental reasons:
\begin{itemize}
  \item Measurements remove phase information, so unitaries that
    achieve the same amplitude magnitudes but with different phases
    during measurement will yield the same measurement outcomes.
  
  \item Two unitaries that are identical to within a permutation of
    the quantum state will also yield the same channel because the
    measure-and-prepare step at the heart of the operation effectively
    removes the effect of such permutations (e.g., a sequence of {\it
      CNOT} gates produces the same overall measure-and-prepare
    channel as the identity gate).
\end{itemize}    
Consequently, any 2-design $\{(p_i, U_i)\}$ can be collapsed to
produce a smaller design $\{(q_j,V_j)\}$, not necessarily a 2-design,
that achieves the same decomposition, where
\begin{itemize}
  \item The set of $V_j$'s is a subset of the $U_i$'s in the 2-design,
    with each $V_j$ defining an equivalence class of $U_i$'s that
    yield identical measure-and-prepare channels: $\Psi_{U_i} =
    \Psi_{V_j}$, $\Psi_{V_j} \not= \Psi_{V_k}$ for $j \not= k$.  
  \item Each $q_j$ is the sum of the $p_i$'s for which $U_i$ produces
    the same channel as $V_j$, $q_j = \sum_{i :
      \Psi_{U_i} = \Psi_{V_j}} p_i$.
\end{itemize}

Just as one can perform an algorithmic search for efficient 2-designs
(e.g.,~\cite{Bravyi2022optimalclifford}), it is also possible to
perform an algorithmic search for efficient unitary designs that
satisfy the equational form of Lowe {\it et al.}, at least for small
numbers of cut wires such as one would encounter in practice.  The
basic approach would be to generate sets of circuits that yield
distinct measure-and-prepare channels $\Psi_U$ and to simultaneously
use linear techniques to find decompositions of the form
\begin{equation}
  \label{eq:linearequation}
  \text{id} = a_0\left(- \Psi_1\right)
  + \sum_{j=1}^{\mathcal{N}} a_j \Psi_{U_j}
  \; .
\end{equation}
Such decompositions can then be directly related back to the
decomposition of Lowe {\it et al.}\ as well as its refinement
presented here.

Various linear techniques can be used to determine whether
coefficients $a_i$ exist that satisfy the decomposition, and to also
calculate their values:
\begin{itemize}
  \item Ordinary least-squares linear regression can be used to find
    coefficients $a_i$ that can be either positive or negative.

  \item Non-negative least-squares linear regression can be used to
    find coefficients $a_i$ that are strictly non-negative.

  \item Linear programming can be used to find coefficients $a_i$
    that minimize the sum of their absolute values $\gamma = \sum_i
    |a_i|$, where $\gamma$ then enters into the Hoeffding bound.
\end{itemize}

The squared errors returned by ordinary and non-negative least-squares
can be used to determine when a sufficient number of
measure-and-prepare channels $\Psi_U$ have been generated so that a
solution then exists that satisfies the constraints of those methods.
In that case, once the squared error falls below a desired tolerance
(e.g., $10^{-12}$), the resulting coefficients can then be used as a
solution to the decomposition.  Linear programming, on the other hand,
requires that the decomposition be introduced as a constraint when
optimizing the coefficients.  As such, linear programming can only be
used once it has been determined, using one of the linear regression
methods, that a sufficient number of measure-and-prepare channels have
been generated.

To simplify the task of applying these linear techniques, we use a
linear-algebra approach to represent the channels $\Psi_1$ and
$\Psi_U$ as 4-dimensional tensors that can then be reshaped into
the 1-dimensional column vectors needed by computer implementations of
these linear techniques.  The construction is provided by the
following lemma.

\begin{lemma}
  \label{lem:tensor}
Any linear transformation $T : \mathbb{C}^{m \times n} \rightarrow
\mathbb{C}^{m \times n}$ over complex matrices $\mathbb{C}^{m \times
  n}$ can be equivalently represented as a 4-dimensional tensor $L^T$
such that, for all $A \in \mathbb{C}^{m \times n}$,
\begin{equation}
  \left(T(A)\right)_{i,j} =  \sum_{\ell,m} L^T_{i, j, \ell,m} A_{\ell, m}
  \; .
\end{equation}
\end{lemma}
\begin{proof}
We begin by noting that the set of $m \times n$ complex matrices is a
vector space in and of itself without having to first ``vectorize''
the matrices.  Specifically, for all matrices $A,B \in \mathbb{C}^{m
  \times n}$, scalars $a,b \in \mathbb{C}$, the zero matrix $0 \in
\mathbb{C}^{m \times n}$, and the unity scalar $1 \in \mathbb{C}$:
\begin{equation}
  \begin{tabular}{L L}
    A + B = B + A & a(bA) = (ab)A \\
    A + (B + C) = (A + B) + C & 1A = A \\
    A + 0 = A & a(A + B) = aA + aB \\
    A + (-A) = 0 & (a + b)A = aA + bA
  \end{tabular}
\end{equation}
In the usual manner, we define $\langle A, B \rangle = \sum_{i,j}
\overline{A}_{i,j} B = \sum_{i,j} A^*_{i,j} B = Tr(A^H B) =
Tr(A^{\dag}B)$, where, depending on your preferred notation,
$\overline{A}_{i,j}$ and $A^*_{i,j}$ are the complex conjugates of
$A$, and $A^H$ and $A^{\dag}$ are the conjugate (i.e., Hermitian)
transposes of $A$.  Then $\langle \cdot , \cdot \rangle$ is an inner
product that is linear in the second argument.  Specifically, for all
matrices $A,B,C \in \mathbb{C}^{m \times n}$, scalars $a,b \in
\mathbb{C}$, and the zero scalar $0 \in \mathbb{R}$:
\begin{equation}
  \begin{tabular}{L}
    \langle A, B \rangle = \overline{\langle B, A \rangle}
    = \langle B, A \rangle^* \\
    \langle A, aB + bC \rangle = a\langle A, B \rangle + b\langle A, C \rangle \\
    \langle A, A \rangle \ge 0 \text{ with equality
      if and only if } A \text{ is the zero matrix}
  \end{tabular}
\end{equation}
Consequently, if the set $\{\beta_s\}$ forms an orthonormal basis for
$\mathbb{C}^{m \times n}$, then for any matrix $A \in \mathbb{C}^{m
  \times n}$,
\begin{equation}
   A = \sum_s \beta_s \langle \beta_s, A \rangle
\end{equation}
Similarly, for any linear transformation $T : \mathbb{C}^{m \times n}
\rightarrow \mathbb{C}^{m \times n}$
\begin{equation}
  T(A) = \sum_s T(\beta_s) \langle \beta_s, A \rangle
  = \sum_r \sum_s \langle \beta_r, T(\beta_s) \rangle
    \beta_r \langle \beta_s, A \rangle
\end{equation}
When written using tensor indices, this equation becomes
\begin{align}
  \left(T(A)\right)_{i,j}
  &= \sum_r \sum_s \langle \beta_r, T(\beta_s) \rangle
    \left(\beta_r\right)_{i,j} \langle \beta_s, A \rangle \notag \\
  &= \sum_r \sum_s \langle \beta_r, T(\beta_s) \rangle
    \left(\beta_r\right)_{i,j} \sum_{\ell,m} \left(\beta^*_s\right)_{\ell,m} A_{\ell,m} \notag \\
  &= \sum_{\ell,m} \left( \sum_r \sum_s \langle \beta_r, T(\beta_s) \rangle
    \left(\beta_r\right)_{i,j} \left(\beta^*_s\right)_{\ell,m} \right)  A_{\ell,m} \notag \\
  &= \sum_{\ell,m} L^T_{i, j, \ell,m} \, A_{\ell, m} \; ,
\end{align}
where
\begin{equation}
  L^T_{i, j, \ell,m} =  \sum_r \sum_s \langle \beta_r, T(\beta_s) \rangle
  \left(\beta_r\right)_{i,j} \left(\beta^*_s\right)_{\ell,m}
  \; .
\end{equation}
Thus, all linear transformations $T$ over complex matrices can be
equivalently represented as 4-dimensional tensors $L^T$, which have
similar multiplication rules as matrices, except that the summations
are taken over two indices instead of just one index.  In terms of
tensor operators, $L^T$ may be expressed as
\begin{equation}
  L^T =  \sum_r \sum_s \langle \beta_r, T(\beta_s) \rangle
  \beta_r \otimes_{\text{outer}} \beta^*_s
  \; ,
\end{equation}
where $\otimes_{\text{outer}}$ is the tensor outer product, not to be
confused with the Kronecker product.  It should be noted that the
choice of orthonormal basis $\{\beta_s\}$ is irrelevant in the above
construction: all choices yield the same tensor $L^T$.
\end{proof}

Because the above construction makes no assumptions regarding the
nature of the linear transformations $T$, the construction can be
directly applied in Liouville space, where the resulting tensors $L^T$
are the tensor representations of superoperators and the density
matrices they operate on do not have to be ``vectorized''
(c.f.,~\cite{Gyamfi2020Liouville}).

The benefit of the tensor representation of superoperators is that it
is then quite straight forward to construct the superoperator tensors
for the channels in Eq.~\ref{eq:linearequation}.  In the case of the
measure-and-prepare channel $\Psi_U$, if we write the equation $\rho'
= \Psi_U(\rho) = \sum_{k=1}^d \braket{k|U^\dag \rho U |k} U
\ketbra{k}{k} U^\dag$ in terms of tensor indices we obtain
\begin{align}
  \rho'_{i,j}
  &= \sum_{k=1}^d \sum_{\ell,m} U^\dag_{k,\ell}\rho_{\ell,m}U_{m,k} U_{i,k}U^\dag_{k,j} \notag \\
  &= \sum_{\ell,m} \left(\sum_{k=1}^d U_{i,k}U^*_{j,k} U^*_{\ell,k}U_{m,k} \right) \rho_{\ell,m} \notag \\
  &= \sum_{\ell,m} L^{\Psi_U}_{i,j,\ell,m} \, \rho_{\ell,m} 
  \; ,
\end{align}
where
\begin{equation}
  L^{\Psi_U}_{i,j,\ell,m}
  = \sum_k U_{i,k} U^*_{j,k} U^*_{\ell,k} U_{m,k}
  \; .
\end{equation}
Similarly, for the depolarizing channel $\rho' = \Psi_1(\rho) = {1
  \over d} \tr(\rho) \mathbbm{1}$ we obtain
\begin{align}
  \rho'_{i,j}
  &= {1 \over d} \sum_{\ell,m} \delta_{\ell,m}\rho_{\ell,m} \delta_{i,j} \notag \\
  &= \sum_{\ell,m} \left( {1 \over d} \delta_{i,j} \delta_{\ell,m} \right) \rho_{\ell,m} \notag \\
  &= \sum_{\ell,m} L^{\Psi_1}_{i,j,\ell,m} \, \rho_{\ell,m} 
  \; ,
\end{align}
where
\begin{equation}
  L^{\Psi_1}_{i,j,\ell,m} = {1 \over d} \, \delta_{i,j} \delta_{\ell,m}
  \; , \qquad
  \delta_{i,j} = \begin{cases}
    1 & i = j \\
    0 & i \not= j
    \end{cases}
  \; .
\end{equation}
For the identity channel,
\begin{equation}
  \rho'_{i,j} = \sum_{\ell,m} \left( \delta_{i,\ell}\delta_{j,m} \right) \rho_{\ell,m}
  = \sum_{\ell,m} L^{\text{id}}_{i,j,\ell,m} \, \rho_{\ell,m} 
  \; , \qquad
  L^{\text{id}}_{i,j,\ell,m} = \delta_{i,\ell}\delta_{j,m}
  \; .
\end{equation}

\begin{algorithm}
    \caption{Measure-and-Prepare Channel Generation for Unitary Designs}
    \label{alg:unitarydesign}
    \hspace*{\algorithmicindent} \textbf{Input:}
    number of qubits $n$, set of one-qubit gates $\mathcal{S}$,
    set of two-qubit gates $\mathcal{D}$, stopping criterion \\
    \hspace*{\algorithmicindent} \textbf{Output:}
    set of unitaries $\mathcal{U}$, set of measure-and-prepare
    superoperator tensors $\mathcal{L}$
    \begin{algorithmic}[1] 
      \State Initialize $\mathcal{U} \gets \{I^{\otimes n}\}, \mathcal{L} \gets \{L^{\Psi_{I^{\otimes n}}}\}$
      \State Initialize $Flag \gets True$
      \While{$Flag$ is $True$ and the stopping criterion has not been met}
          \State $Flag \gets False$
          \For{$U \in \mathcal{U}$}
              \For{$G \in \mathcal{S}$ and $k$ from 1 to $n$}
                      \State $V \gets (I^{\otimes k-1} \otimes G \otimes I^{\otimes n-k})U$
                      \If{$L^{\Psi_V} \notin \mathcal{L}$}
                          \State $\mathcal{U} \gets \mathcal{U} \cup \{G\}$
                          \State $\mathcal{L} \gets \mathcal{L} \cup \{L^{\Psi_V}\}$
                          \State $Flag \gets True$
                      \EndIf
              \EndFor
              \For{$G \in \mathcal{D}$}
                  \State $V \gets DU$
                  \If{$L^{\Psi_V} \notin \mathcal{L}$}
                      \State $\mathcal{U} \gets \mathcal{U} \cup \{G\}$
                      \State $\mathcal{L} \gets \mathcal{L} \cup \{L^{\Psi_V}\}$
                      \State $Flag \gets True$
                  \EndIf
              \EndFor
          \EndFor
      \EndWhile
      \State \textbf{return} $\mathcal{U},\mathcal{L}$
    \end{algorithmic}
\end{algorithm}

The above definitions can be readily implemented using numerical
packages that support tensor computations.  The resulting
4-dimensional superoperator tensors can then be flattened into
1-dimensional vectors for use in solving Eq.~\ref{eq:linearequation}.
Specifically, flattened versions of $-L^{\Psi_1}$ and $L^{\Psi_{U_j}}$
for each value of $j$ would be assembled to form the columns of an
``$A$'' matrix, $L^{\text{id}}$ would become a ``$b$'' vector, and the
coefficients in Eq.~\ref{eq:linearequation} would be obtained by
solving linear equations of the form $Ax=b$.

Algorithm~\ref{alg:unitarydesign} defines a method for constructing
sets of unitaries and their corresponding measure-and-prepare channels
that can then be used to solve Eq.~\ref{eq:linearequation} as
described above.  A desired subgroup of unitaries from which to
construct measure-and-prepare channels can be specified by providing
the generating elements of the subgroup as input.  In practice, when
applying the resulting unitary designs in actual circuits, {\it SWAP}
gates may be needed to move the wires that are to be cut to a
connected subgraph of qubits on a target quantum device.  Once moved,
it would be preferable not to apply any additional {\it SWAP} gates
when performing the measure-and-prepare operations needed for wire
cutting.  Algorithm~\ref{alg:unitarydesign} accommodates this need by
requiring that two sets of generating elements be provided: a set of
single-qubit gates $\mathcal{S}$ that can be applied to any qubit
being cut, and a set of two-qubit gates $\mathcal{D}$ that are
qubit-specific so that the connectivity graph of the target quantum
device can be respected by supplying only those qubit combinations
that are natively supported by the hardware.  In this manner,
Algorithm~\ref{alg:unitarydesign} can be used to generate customized
unitary designs for specific regions of specific target devices.

A third input to Algorithm~\ref{alg:unitarydesign} is a stopping
criterion.  As described earlier, the squared errors returned by
ordinary and non-negative least-squares can be used to determine when
a sufficient number of measure-and-prepare channels $\Psi_U$ have been
generated so that a solution then exists that satisfies the
constraints of those methods.  Once the squared error falls below a
desired tolerance (e.g., $10^{-12}$), the resulting coefficients can
then be used as a solution to the decomposition.  This is then the
stopping criterion used in the experiments conducted thus far, using
either ordinary or non-negative least-squares regression, depending on
the experiment.

Our initial experiments have focused on Clifford designs without
device connectivity constraints.  The single-qubit generating elements
employed were thus Hadamard gates $H$ and phase gates $S$, while the
two-qubit generators were {\it CNOT} gates across all pairs of qubits.

In our experience thus far in generating Clifford designs that satisfy
Eq.~\ref{eq:linearequation}, off-the-shelf ordinary least-squares
regression (i.e., numpy.linalg.lstsq) identifies decompositions that
require the fewest Clifford circuits, but the resulting $\gamma$'s
(i.e., the sums of the absolute values of the regression coefficients)
always seem to be greater than those obtained using non-negative
least-squares and linear programming.

Off-the-shelf non-negative least-squares regression (i.e.,
scipy.optimize.nnls) seems to identify decompositions that require
more Clifford circuits than ordinary least-squares, but the resulting
$\gamma$'s are then identical to the $\gamma$'s discussed earlier for
2-designs, but using far fewer Clifford circuits.  The resulting
Clifford circuits and their coefficients are presented in Appendix~A
for 1, 2, 3, and 4 wire cuts.  The numbers of unitaries in these
designs are 3, 5, 28, and 136, respectively, as compared to the sizes
of the corresponding Clifford groups, which are 8, 192, 92,160, and
743,178,240~\cite{oeis:A00395}.  Appendix~A includes the associated
probabilities of the resulting sets of Clifford circuits.  All of
these unitary designs satisfy both the original wire-cutting
decomposition of Lowe {\it et al.}\ as well as the improved
decomposition presented in the previous section.

Off-the-shelf linear programming (i.e., scipy.optimize.linprog) was
only applied once a sufficient number of measure-and-prepare channels
were generated for a solution to be obtained using non-negative
least-squares regression.  The linear programming solutions did not
improve the resulting $\gamma$'s, but the solutions that were found
were much larger in size than those obtained using non-negative least
squares.  Because these results are therefore inferior to those
obtained using non-negative least-squares in terms of their sizes,
they are not reported.  Interestingly, all solutions found using
linear programming also had non-negative coefficients, even though no
constraints were added forcing the coefficients to be non-negative.

\section{Comparisons with Harada et al.}
One very interesting aspect of the decomposition just published by
Harada {\it et al.}~\cite{Harada2023optimalwirecutting} is that the
unitary designs they consider contain only $d+1 = 2^n + 1$ unitaries,
including the identity.  Of the four designs presented in Appendix~A,
the first design for one qubit matches the design presented
in~\cite{Harada2023optimalwirecutting}.  The design for two qubits
matches the size of the design reported
in~\cite{Harada2023optimalwirecutting}, but it does not include the
identity circuit.  The other two designs contain 28 and 136 unitaries
for three and four-qubit wire cuts, respectively, versus the 9 and 17
unitaries that would be obtained by considering transformations from
the computational bases to mutually unbiased bases.  Because it is
easy to constrain least-squares methods to include specific
regressors, we re-executed our experiments to include the identity
circuit and obtained the unitary designs presented in Appendix~B.  The
resulting designs for three and four wire cuts then contained 17 and
122 unitaries, which is a mild improvement.  However, the numbers of
gates required to implement the unitaries then increased, which
indicates that a trade-off must therefore exist between the number of
unitaries that appear in a measure-and-prepare decomposition and the
number of gates needed to implement those unitaries.

As mentioned in the Introduction, the unitary design presented by
Harada {\it et al.}~\cite{Harada2023optimalwirecutting} also satisfies
the decomposition presented here.  From a practical point of view, the
question then becomes which decomposition is easier to implement on a
target quantum device.  Algorithm~\ref{alg:altmontecarlo} was
introduced in Section~2 to provide an alternative means of preparing
the mixed states that are needed so they can be more easily
implemented on existing hardware.  A similar type of mixed state
appears in the decomposition of Harada {\it et al.} (see
Eq.~\ref{eq:HaradaMixedState}).  A similar approach could potentially
be used to handle this mixed state as was used in developing
Algorithm~\ref{alg:altmontecarlo}.  The resulting algorithm, though,
would likely be very similar to Algorithm~\ref{alg:altmontecarlo},
since both decompositions are basically algebraic variations of each
other.

\section{Conclusions}
With the modified decomposition presented here, the $\gamma$ factor
that determines the rate of convergence of estimation errors for
expected values of observables is reduced from $\gamma = 2^{n+1} + 1$
to $\gamma = 2^{n+1} - 1$, which matches the optimum $\gamma$ reported
by Brenner {\it et al.}~\cite{Brenner2023optimalwirecutting}, but
without using ancilla qubits in the wire cutting.  This $\gamma$ also
matches that of the recently published decomposition reported by
Harada {\it et al.}~\cite{Harada2023optimalwirecutting}.

The algorithms explored for constructing improved unitary designs
produce much, much smaller sets of unitaries than 2-designs, which
then makes both the original decomposition of Lowe {\it
  al.}~\cite{Lowe2023fastquantumcircuit}, and the improved
decomposition presented here, much more practical for real-world
application.  The algorithm for generating candidate
measure-and-prepare channels allows custom sets of generating gates to
be supplied as input, which then enables custom unitary designs to be
generated that respect qubit connectivity patterns in specific regions
of specific target devices.  This capability could potentially be
leveraged as an optimization when transpiling cut circuits to specific
target devices.  However, the sizes of the resulting unitary designs
can be further improved as demonstrated by the unitary designs
reported by Harada {\it et al.}~\cite{Harada2023optimalwirecutting}.
Further research is therefore needed to produce custom unitary designs
of comparable size that respect qubit connectivity patterns of target
quantum devices.

\section*{Acknowledgments}
We are very grateful to David Sutter and Stefan W\"{o}rner for their
comments and recommendations for improving earlier versions of this
manuscript.

\clearpage
\appendix
\section{Appendix}

The following are the unitary designs obtained for 1, 2, 3, and 4 wire
cuts using non-negative least-squares regression.  The sums of these
coefficients match the value of $\gamma = 2^{n+1} + 1$ defined in the
decomposition of Lowe {\it et al.}~\cite{Lowe2023fastquantumcircuit}.
The probabilities define the corresponding unitary designs.  All of
these designs satisfy both the original decomposition of Lowe {\it et
  al.}\ and the improved decomposition presented in this paper.

\subsection{One Wire Cut, 3 Unitaries}
\begin{longtable}{c c l c}
\textbf{Coefficient} & \textbf{Probabilty} & \textbf{PseudoQASM Circuit} & \textbf{Num Gates} \\
\endhead
2.000000 & & Complete Depolarization & \\
1.000000 & 0.33333333 & Identity Circuit & 0 \\
1.000000 & 0.33333333 & h 0; & 1 \\
1.000000 & 0.33333333 & h 0; s 0; & 2 \\
\end{longtable}

\subsection{Two Wire Cuts, 5 Unitaries}
\begin{longtable}{c c l c}
\textbf{Coefficient} & \textbf{Probabilty} & \textbf{PseudoQASM Circuit} & \textbf{Num Gates} \\
\endhead
4.000000 & & Complete Depolarization & \\
1.000000 & 0.20000000 & h 0; h 1; & 2 \\
1.000000 & 0.20000000 & h 0; s 0; & 2 \\
1.000000 & 0.20000000 & h 1; s 1; & 2 \\
1.000000 & 0.20000000 & h 0; s 0; cx 0, 1; & 3 \\
1.000000 & 0.20000000 & h 0; cx 0, 1; h 0; & 3 \\
\end{longtable}

\subsection{Three Wire Cuts, 28 Unitaries}
\begin{longtable}{c c l c}
\textbf{Coefficient} & \textbf{Probabilty} & \textbf{PseudoQASM Circuit} & \textbf{Num Gates} \\
\endhead
8.000000 & & Complete Depolarization & \\
0.500000 & 0.05555556 & h 0; h 1; h 2; & 3 \\
0.125000 & 0.01388889 & h 0; h 1; cx 1, 2; & 3 \\
0.500000 & 0.05555556 & h 0; s 0; cx 0, 2; & 3 \\
0.625000 & 0.06944444 & h 0; cx 0, 1; h 0; & 3 \\
0.125000 & 0.01388889 & h 0; h 1; s 0; s 1; & 4 \\
0.375000 & 0.04166667 & h 0; h 1; s 0; cx 0, 2; & 4 \\
0.250000 & 0.02777778 & h 0; h 1; s 0; cx 1, 0; & 4 \\
0.125000 & 0.01388889 & h 0; h 2; s 0; s 2; & 4 \\
0.250000 & 0.02777778 & h 0; h 2; s 0; cx 2, 0; & 4 \\
0.125000 & 0.01388889 & h 0; s 0; cx 0, 2; h 2; & 4 \\
0.250000 & 0.02777778 & h 1; h 2; s 1; s 2; & 4 \\
0.125000 & 0.01388889 & h 1; h 2; s 1; cx 2, 1; & 4 \\
0.625000 & 0.06944444 & h 1; s 1; cx 1, 2; h 2; & 4 \\
0.375000 & 0.04166667 & h 0; h 1; s 0; s 1; cx 1, 2; & 5 \\
0.125000 & 0.01388889 & h 0; h 1; s 0; cx 0, 2; h 0; & 5 \\
0.375000 & 0.04166667 & h 0; h 1; s 0; cx 1, 2; h 1; & 5 \\
0.625000 & 0.06944444 & h 0; h 1; s 1; cx 0, 2; h 0; & 5 \\
0.250000 & 0.02777778 & h 0; h 1; s 1; cx 0, 2; cx 1, 0; & 5 \\
0.375000 & 0.04166667 & h 0; h 1; s 1; cx 1, 2; h 1; & 5 \\
0.250000 & 0.02777778 & h 0; h 1; cx 0, 2; h 0; cx 1, 0; & 5 \\
0.250000 & 0.02777778 & h 0; h 1; cx 0, 2; h 0; cx 1, 2; & 5 \\
0.625000 & 0.06944444 & h 0; h 2; s 0; s 2; cx 0, 1; & 5 \\
0.375000 & 0.04166667 & h 0; h 2; s 0; cx 0, 1; h 0; & 5 \\
0.125000 & 0.01388889 & h 0; h 2; s 0; cx 0, 1; h 1; & 5 \\
0.375000 & 0.04166667 & h 0; h 2; s 0; cx 2, 0; cx 0, 1; & 5 \\
0.125000 & 0.01388889 & h 0; s 0; cx 0, 1; h 0; cx 0, 2; & 5 \\
0.500000 & 0.05555556 & h 0; s 0; cx 0, 1; h 1; cx 1, 2; & 5 \\
0.250000 & 0.02777778 & h 0; cx 0, 1; h 0; cx 0, 2; cx 1, 0; & 5 \\
\end{longtable}

\subsection{Four Wire Cuts, 136 Unitaries}
\begin{longtable}{c c l c}
\textbf{Coefficient} & \textbf{Probabilty} & \textbf{PseudoQASM Circuit} & \textbf{Num Gates} \\
\endhead
16.000000 & & Complete Depolarization & \\
0.077381 & 0.00455182 & h 0; h 1; s 0; cx 0, 3; h 3; & 5 \\
0.177954 & 0.01046790 & h 0; h 1; s 1; cx 0, 2; cx 1, 0; & 5 \\
0.117719 & 0.00692462 & h 0; h 2; s 0; cx 0, 3; h 0; & 5 \\
0.070027 & 0.00411922 & h 0; h 2; s 0; cx 0, 3; cx 2, 0; & 5 \\
0.018012 & 0.00105952 & h 0; h 2; s 0; cx 0, 3; cx 2, 3; & 5 \\
0.089097 & 0.00524098 & h 0; h 2; s 2; cx 0, 1; cx 2, 3; & 5 \\
0.071553 & 0.00420902 & h 0; h 3; s 0; cx 0, 1; h 1; & 5 \\
0.105949 & 0.00623232 & h 1; h 2; s 2; cx 2, 3; h 3; & 5 \\
0.072323 & 0.00425430 & h 0; h 1; h 2; h 3; s 0; cx 1, 0; & 6 \\
0.084937 & 0.00499627 & h 0; h 1; h 2; h 3; s 1; cx 2, 1; & 6 \\
0.054453 & 0.00320312 & h 0; h 1; h 2; h 3; s 2; cx 3, 2; & 6 \\
0.048801 & 0.00287065 & h 0; h 1; h 2; s 0; cx 2, 0; cx 2, 3; & 6 \\
0.182626 & 0.01074273 & h 0; h 1; h 3; s 0; s 1; cx 3, 0; & 6 \\
0.057030 & 0.00335469 & h 0; h 1; h 3; s 0; s 3; cx 0, 2; & 6 \\
0.081442 & 0.00479070 & h 0; h 1; h 3; s 0; s 3; cx 1, 0; & 6 \\
0.178872 & 0.01052189 & h 0; h 1; h 3; s 1; s 3; cx 1, 2; & 6 \\
0.023948 & 0.00140871 & h 0; h 1; s 0; s 1; cx 1, 2; h 1; & 6 \\
0.099682 & 0.00586366 & h 0; h 1; s 0; cx 0, 2; h 0; cx 1, 2; & 6 \\
0.055007 & 0.00323569 & h 0; h 1; s 0; cx 0, 2; h 2; cx 1, 0; & 6 \\
0.065589 & 0.00385817 & h 0; h 1; s 0; cx 0, 2; cx 1, 0; h 0; & 6 \\
0.026757 & 0.00157396 & h 0; h 1; s 0; cx 0, 2; cx 1, 0; cx 1, 3; & 6 \\
0.137451 & 0.00808535 & h 0; h 1; s 0; cx 0, 2; cx 1, 2; h 2; & 6 \\
0.307929 & 0.01811345 & h 0; h 1; s 0; cx 0, 3; h 0; cx 0, 2; & 6 \\
0.108920 & 0.00640705 & h 0; h 1; s 0; cx 0, 3; h 0; cx 1, 3; & 6 \\
0.034429 & 0.00202525 & h 0; h 1; s 0; cx 0, 3; h 3; cx 1, 0; & 6 \\
0.156286 & 0.00919332 & h 0; h 1; s 0; cx 0, 3; cx 1, 0; cx 1, 2; & 6 \\
0.069658 & 0.00409751 & h 0; h 1; s 1; cx 0, 2; cx 0, 3; cx 1, 0; & 6 \\
0.300465 & 0.01767439 & h 0; h 1; s 1; cx 0, 2; cx 0, 3; cx 1, 2; & 6 \\
0.342935 & 0.02017266 & h 0; h 1; s 1; cx 0, 3; cx 1, 0; h 1; & 6 \\
0.351542 & 0.02067894 & h 0; h 1; s 1; cx 1, 2; h 1; cx 1, 3; & 6 \\
0.231054 & 0.01359143 & h 0; h 1; s 1; cx 1, 3; h 1; cx 1, 2; & 6 \\
0.056278 & 0.00331044 & h 0; h 1; cx 0, 2; h 0; cx 0, 3; cx 1, 3; & 6 \\
0.054619 & 0.00321287 & h 0; h 1; cx 0, 2; h 0; cx 1, 0; cx 1, 3; & 6 \\
0.055383 & 0.00325779 & h 0; h 1; cx 0, 2; h 0; cx 1, 3; cx 2, 1; & 6 \\
0.128292 & 0.00754657 & h 0; h 1; cx 0, 2; cx 0, 3; h 3; cx 1, 2; & 6 \\
0.068560 & 0.00403294 & h 0; h 1; cx 0, 2; cx 0, 3; cx 1, 2; h 1; & 6 \\
0.132793 & 0.00781133 & h 0; h 1; cx 0, 3; h 0; cx 1, 0; cx 1, 2; & 6 \\
0.171234 & 0.01007257 & h 0; h 1; cx 0, 3; cx 1, 0; cx 1, 2; h 1; & 6 \\
0.111486 & 0.00655800 & h 0; h 2; h 3; s 0; s 3; cx 0, 1; & 6 \\
0.049308 & 0.00290046 & h 0; h 2; h 3; s 0; s 3; cx 2, 0; & 6 \\
0.067551 & 0.00397362 & h 0; h 2; h 3; s 0; cx 2, 0; cx 0, 1; & 6 \\
0.067291 & 0.00395829 & h 0; h 2; s 0; s 2; cx 0, 3; h 0; & 6 \\
0.278136 & 0.01636096 & h 0; h 2; s 0; cx 0, 1; h 0; cx 0, 3; & 6 \\
0.172267 & 0.01013333 & h 0; h 2; s 0; cx 0, 1; h 1; cx 1, 3; & 6 \\
0.040293 & 0.00237016 & h 0; h 2; s 0; cx 0, 3; cx 2, 0; cx 0, 1; & 6 \\
0.107591 & 0.00632888 & h 0; h 2; s 0; cx 0, 3; cx 2, 3; h 2; & 6 \\
0.232530 & 0.01367821 & h 0; h 2; s 2; cx 0, 3; h 0; cx 2, 3; & 6 \\
0.088464 & 0.00520376 & h 0; h 2; cx 0, 1; cx 0, 3; cx 2, 3; h 3; & 6 \\
0.146617 & 0.00862451 & h 0; h 2; cx 0, 3; h 0; cx 2, 0; cx 0, 3; & 6 \\
0.095580 & 0.00562235 & h 0; h 3; s 0; cx 0, 2; h 0; cx 0, 1; & 6 \\
0.012711 & 0.00074768 & h 1; h 2; h 3; s 1; s 2; cx 3, 2; & 6 \\
0.036036 & 0.00211979 & h 1; h 2; h 3; s 1; s 3; cx 2, 1; & 6 \\
0.180186 & 0.01059920 & h 1; h 2; s 1; cx 1, 3; h 1; cx 2, 3; & 6 \\
0.226670 & 0.01333355 & h 1; h 2; s 1; cx 1, 3; h 3; cx 2, 1; & 6 \\
0.176104 & 0.01035906 & h 1; h 2; s 2; cx 1, 3; cx 2, 1; h 2; & 6 \\
0.193634 & 0.01139026 & h 0; h 1; h 2; s 0; s 1; cx 0, 3; cx 2, 0; & 7 \\
0.311273 & 0.01831017 & h 0; h 1; h 2; s 0; s 1; cx 2, 3; h 2; & 7 \\
0.138706 & 0.00815915 & h 0; h 1; h 2; s 0; s 2; cx 0, 3; cx 1, 3; & 7 \\
0.125152 & 0.00736191 & h 0; h 1; h 2; s 0; s 2; cx 1, 0; cx 1, 3; & 7 \\
0.104305 & 0.00613559 & h 0; h 1; h 2; s 0; s 2; cx 1, 0; cx 2, 3; & 7 \\
0.112402 & 0.00661191 & h 0; h 1; h 2; s 0; s 2; cx 1, 3; cx 2, 1; & 7 \\
0.149947 & 0.00882044 & h 0; h 1; h 2; s 0; cx 0, 3; cx 1, 3; cx 2, 1; & 7 \\
0.104210 & 0.00612999 & h 0; h 1; h 2; s 0; cx 1, 0; cx 0, 3; cx 2, 3; & 7 \\
0.111984 & 0.00658729 & h 0; h 1; h 2; s 0; cx 2, 0; cx 0, 3; h 0; & 7 \\
0.104404 & 0.00614139 & h 0; h 1; h 2; s 0; cx 2, 0; cx 2, 3; h 2; & 7 \\
0.203606 & 0.01197682 & h 0; h 1; h 2; s 1; s 2; cx 0, 3; h 0; & 7 \\
0.099142 & 0.00583188 & h 0; h 1; h 2; s 1; s 2; cx 1, 3; h 3; & 7 \\
0.132070 & 0.00776880 & h 0; h 1; h 2; s 2; cx 0, 3; h 0; cx 2, 0; & 7 \\
0.069135 & 0.00406678 & h 0; h 1; h 2; s 2; cx 0, 3; cx 1, 0; cx 2, 0; & 7 \\
0.083391 & 0.00490533 & h 0; h 1; h 3; s 0; s 1; cx 1, 2; h 1; & 7 \\
0.208424 & 0.01226026 & h 0; h 1; h 3; s 0; s 1; cx 1, 2; h 2; & 7 \\
0.018317 & 0.00107749 & h 0; h 1; h 3; s 0; s 1; cx 1, 2; cx 3, 0; & 7 \\
0.100713 & 0.00592429 & h 0; h 1; h 3; s 0; s 1; cx 3, 0; cx 0, 1; & 7 \\
0.096150 & 0.00565587 & h 0; h 1; h 3; s 0; s 1; cx 3, 0; cx 0, 2; & 7 \\
0.003405 & 0.00020032 & h 0; h 1; h 3; s 0; s 3; cx 1, 0; cx 0, 2; & 7 \\
0.260562 & 0.01532716 & h 0; h 1; h 3; s 0; s 3; cx 1, 2; h 1; & 7 \\
0.121145 & 0.00712619 & h 0; h 1; h 3; s 0; cx 1, 0; cx 1, 2; h 2; & 7 \\
0.039427 & 0.00231925 & h 0; h 1; h 3; s 0; cx 3, 0; cx 0, 2; cx 1, 0; & 7 \\
0.036867 & 0.00216862 & h 0; h 1; h 3; s 1; cx 0, 2; cx 1, 0; h 1; & 7 \\
0.197091 & 0.01159358 & h 0; h 1; s 0; s 1; cx 0, 2; h 0; cx 0, 1; & 7 \\
0.243278 & 0.01431048 & h 0; h 1; s 0; s 1; cx 0, 2; h 2; cx 2, 1; & 7 \\
0.057754 & 0.00339732 & h 0; h 1; s 0; cx 0, 2; h 0; cx 0, 3; cx 1, 0; & 7 \\
0.096286 & 0.00566388 & h 0; h 1; s 0; cx 0, 2; h 0; cx 1, 0; cx 0, 3; & 7 \\
0.051346 & 0.00302035 & h 0; h 1; s 0; cx 0, 2; h 0; cx 1, 2; cx 1, 3; & 7 \\
0.023306 & 0.00137093 & h 0; h 1; s 0; cx 0, 2; h 0; cx 1, 2; cx 2, 3; & 7 \\
0.219133 & 0.01289015 & h 0; h 1; s 0; cx 0, 2; h 0; cx 1, 3; h 1; & 7 \\
0.118390 & 0.00696409 & h 0; h 1; s 0; cx 0, 2; h 0; cx 2, 3; cx 1, 2; & 7 \\
0.342627 & 0.02015455 & h 0; h 1; s 0; cx 0, 2; h 2; cx 0, 3; cx 1, 0; & 7 \\
0.047876 & 0.00281621 & h 0; h 1; s 0; cx 0, 2; h 2; cx 0, 3; cx 1, 3; & 7 \\
0.380592 & 0.02238779 & h 0; h 1; s 0; cx 0, 2; h 2; cx 1, 2; cx 2, 3; & 7 \\
0.224981 & 0.01323416 & h 0; h 1; s 0; cx 0, 2; h 2; cx 1, 3; h 1; & 7 \\
0.055010 & 0.00323590 & h 0; h 1; s 0; cx 0, 2; cx 0, 3; cx 1, 2; h 1; & 7 \\
0.061342 & 0.00360834 & h 0; h 1; s 0; cx 0, 2; cx 1, 0; cx 1, 3; h 1; & 7 \\
0.009913 & 0.00058309 & h 0; h 1; s 0; cx 0, 2; cx 1, 0; cx 1, 3; h 3; & 7 \\
0.006922 & 0.00040716 & h 0; h 1; s 0; cx 0, 2; cx 1, 2; cx 1, 3; h 1; & 7 \\
0.113841 & 0.00669654 & h 0; h 1; s 0; cx 0, 2; cx 1, 2; cx 1, 3; h 3; & 7 \\
0.249738 & 0.01469046 & h 0; h 1; s 0; cx 0, 3; h 3; cx 1, 0; cx 0, 2; & 7 \\
0.242043 & 0.01423783 & h 0; h 1; s 0; cx 0, 3; h 3; cx 1, 2; cx 1, 3; & 7 \\
0.020993 & 0.00123491 & h 0; h 1; s 0; cx 1, 0; cx 1, 2; h 1; cx 1, 3; & 7 \\
0.149748 & 0.00880871 & h 0; h 1; s 0; cx 1, 0; cx 1, 2; h 1; cx 2, 3; & 7 \\
0.001248 & 0.00007342 & h 0; h 1; s 1; cx 0, 2; h 0; cx 1, 0; cx 0, 3; & 7 \\
0.337875 & 0.01987501 & h 0; h 1; s 1; cx 0, 2; h 0; cx 1, 3; h 1; & 7 \\
0.332991 & 0.01958770 & h 0; h 1; s 1; cx 0, 2; h 0; cx 1, 3; h 3; & 7 \\
0.139003 & 0.00817666 & h 0; h 1; s 1; cx 0, 2; cx 0, 3; cx 1, 3; h 1; & 7 \\
0.002701 & 0.00015886 & h 0; h 1; s 1; cx 0, 2; cx 1, 0; h 1; cx 1, 3; & 7 \\
0.133693 & 0.00786431 & h 0; h 1; s 1; cx 0, 2; cx 1, 3; h 1; cx 1, 0; & 7 \\
0.151284 & 0.00889903 & h 0; h 1; s 1; cx 0, 2; cx 1, 3; h 3; cx 3, 0; & 7 \\
0.116149 & 0.00683229 & h 0; h 1; s 1; cx 0, 3; h 0; cx 1, 2; h 2; & 7 \\
0.052431 & 0.00308417 & h 0; h 1; s 1; cx 0, 3; cx 1, 0; cx 1, 2; h 2; & 7 \\
0.078384 & 0.00461084 & h 0; h 1; s 1; cx 0, 3; cx 1, 2; h 1; cx 1, 0; & 7 \\
0.085250 & 0.00501469 & h 0; h 1; cx 0, 2; h 0; cx 0, 3; cx 1, 3; cx 2, 0; & 7 \\
0.156046 & 0.00917916 & h 0; h 1; cx 0, 2; h 0; cx 0, 3; cx 2, 0; cx 1, 2; & 7 \\
0.200294 & 0.01178202 & h 0; h 1; cx 0, 2; h 0; cx 1, 0; cx 0, 2; cx 0, 3; & 7 \\
0.097268 & 0.00572162 & h 0; h 1; cx 0, 2; h 0; cx 1, 0; cx 0, 2; cx 1, 3; & 7 \\
0.134494 & 0.00791138 & h 0; h 1; cx 0, 2; h 0; cx 1, 0; cx 0, 2; cx 2, 3; & 7 \\
0.146746 & 0.00863209 & h 0; h 1; cx 0, 2; cx 0, 3; cx 1, 3; h 1; cx 0, 1; & 7 \\
0.184532 & 0.01085485 & h 0; h 1; cx 0, 2; cx 1, 0; cx 1, 3; h 1; cx 1, 0; & 7 \\
0.040634 & 0.00239023 & h 0; h 1; cx 0, 3; h 0; cx 1, 0; cx 0, 3; cx 1, 2; & 7 \\
0.045664 & 0.00268612 & h 0; h 2; h 3; s 0; s 2; cx 0, 1; h 0; & 7 \\
0.001225 & 0.00007205 & h 0; h 2; h 3; s 0; s 2; cx 0, 1; h 1; & 7 \\
0.035090 & 0.00206411 & h 0; h 2; h 3; s 0; s 2; cx 0, 1; cx 3, 2; & 7 \\
0.111340 & 0.00654940 & h 0; h 2; h 3; s 0; cx 2, 0; cx 0, 1; h 0; & 7 \\
0.248069 & 0.01459231 & h 0; h 2; h 3; s 0; cx 3, 0; cx 0, 1; h 0; & 7 \\
0.221859 & 0.01305053 & h 0; h 2; h 3; s 2; s 3; cx 0, 1; h 0; & 7 \\
0.097355 & 0.00572679 & h 0; h 2; s 0; s 2; cx 0, 1; h 1; cx 1, 3; & 7 \\
0.073219 & 0.00430699 & h 0; h 2; s 0; s 2; cx 0, 3; h 0; cx 0, 2; & 7 \\
0.041301 & 0.00242946 & h 0; h 2; s 0; s 2; cx 0, 3; h 3; cx 2, 0; & 7 \\
0.017849 & 0.00104993 & h 0; h 2; s 0; s 2; cx 2, 3; h 2; cx 2, 0; & 7 \\
0.058538 & 0.00344344 & h 0; h 2; s 0; s 2; cx 2, 3; h 3; cx 0, 2; & 7 \\
0.373068 & 0.02194518 & h 0; h 2; s 0; cx 0, 1; h 0; cx 1, 3; cx 2, 3; & 7 \\
0.153232 & 0.00901365 & h 0; h 2; s 0; cx 0, 1; h 1; cx 2, 3; h 2; & 7 \\
0.169538 & 0.00997284 & h 0; h 2; s 2; cx 0, 1; cx 0, 3; cx 2, 3; h 2; & 7 \\
0.250401 & 0.01472949 & h 0; h 2; s 2; cx 0, 1; cx 2, 3; h 3; cx 0, 3; & 7 \\
0.052500 & 0.00308822 & h 0; h 2; cx 0, 1; cx 0, 3; cx 2, 3; h 2; cx 0, 2; & 7 \\
0.123896 & 0.00728798 & h 1; h 2; h 3; s 1; s 2; cx 3, 2; cx 2, 1; & 7 \\
0.138447 & 0.00814395 & h 1; h 2; s 1; s 2; cx 1, 3; h 3; cx 3, 2; & 7 \\
\end{longtable}

\section{Appendix}

The following are the unitary designs obtained for 1, 2, 3, and 4 wire
cuts using non-negative least-squares regression with the constraint
that the identity circuit must be included in the unitary design.  The
sums of these coefficients match the value of $\gamma = 2^{n+1} + 1$
defined in the decomposition of Lowe {\it et
  al.}~\cite{Lowe2023fastquantumcircuit}.  The probabilities define
the corresponding unitary designs.  All of these designs satisfy the
original decomposition of Lowe {\it et al.}\ and the improved
decomposition presented in this paper.

\subsection{One Wire Cut, 3 Unitaries}
\begin{longtable}{c c l c}
\textbf{Coefficient} & \textbf{Probabilty} & \textbf{PseudoQASM Circuit} & \textbf{Num Gates} \\
\endhead
2.000000 & & Complete Depolarization & \\
1.000000 & 0.33333333 & Identity Circuit & 0 \\
1.000000 & 0.33333333 & h 0; & 1 \\
1.000000 & 0.33333333 & h 0; s 0; & 2 \\
\end{longtable}

\subsection{Two Wire Cuts, 5 Unitaries}
\begin{longtable}{c c l c}
\textbf{Coefficient} & \textbf{Probabilty} & \textbf{PseudoQASM Circuit} & \textbf{Num Gates} \\
\endhead
4.000000 & & Complete Depolarization & \\
1.000000 & 0.20000000 & Identity Circuit & 0 \\
1.000000 & 0.20000000 & h 0; h 1; & 2 \\
1.000000 & 0.20000000 & h 0; h 1; s 0; s 1; & 4 \\
1.000000 & 0.20000000 & h 0; s 0; cx 0, 1; h 0; & 4 \\
1.000000 & 0.20000000 & h 0; s 0; cx 0, 1; h 1; & 4 \\
\end{longtable}

\subsection{Three Wire Cuts, 17 Unitaries}
\begin{longtable}{c c l c}
\textbf{Coefficient} & \textbf{Probabilty} & \textbf{PseudoQASM Circuit} & \textbf{Num Gates} \\
\endhead
8.000000 & & Complete Depolarization & \\
1.000000 & 0.11111111 & Identity Circuit & 0 \\
0.500000 & 0.05555556 & h 0; h 2; cx 0, 1; h 0; & 4 \\
0.500000 & 0.05555556 & h 0; h 1; h 2; s 0; s 2; & 5 \\
0.500000 & 0.05555556 & h 0; h 1; h 2; s 0; cx 1, 0; & 5 \\
0.500000 & 0.05555556 & h 0; h 1; s 0; cx 0, 2; h 2; & 5 \\
0.500000 & 0.05555556 & h 0; h 1; s 0; cx 1, 2; h 1; & 5 \\
0.500000 & 0.05555556 & h 0; h 1; s 1; cx 1, 2; h 1; & 5 \\
0.500000 & 0.05555556 & h 0; h 1; s 1; cx 1, 2; h 2; & 5 \\
0.500000 & 0.05555556 & h 0; h 1; h 2; s 0; s 1; cx 2, 0; & 6 \\
0.500000 & 0.05555556 & h 0; h 1; h 2; s 0; cx 1, 0; cx 2, 0; & 6 \\
0.500000 & 0.05555556 & h 0; h 1; s 0; s 1; cx 0, 2; h 0; & 6 \\
0.500000 & 0.05555556 & h 0; h 1; s 0; cx 0, 2; h 0; cx 1, 2; & 6 \\
0.500000 & 0.05555556 & h 0; h 1; s 0; cx 0, 2; cx 1, 0; h 1; & 6 \\
0.500000 & 0.05555556 & h 0; h 1; s 0; cx 1, 0; cx 0, 2; h 2; & 6 \\
0.500000 & 0.05555556 & h 0; h 1; s 1; cx 0, 2; h 0; cx 1, 0; & 6 \\
0.500000 & 0.05555556 & h 0; h 1; cx 0, 2; h 0; cx 1, 0; cx 0, 2; & 6 \\
0.500000 & 0.05555556 & h 0; h 2; s 0; s 2; cx 0, 1; h 1; & 6 \\
\end{longtable}

\subsection{Four Wire Cuts, 122 Unitaries}
\begin{longtable}{c c l c}
\textbf{Coefficient} & \textbf{Probabilty} & \textbf{PseudoQASM Circuit} & \textbf{Num Gates} \\
\endhead
16.000000 & & Complete Depolarization & \\
1.000000 & 0.05882353 & Identity Circuit & 0 \\
0.214221 & 0.01260126 & h 0; h 1; h 2; s 0; cx 0, 3; h 0; cx 1, 3; & 7 \\
0.003864 & 0.00022732 & h 0; h 1; h 2; s 0; cx 1, 0; cx 1, 3; h 1; & 7 \\
0.055663 & 0.00327429 & h 0; h 1; h 2; s 0; cx 1, 0; cx 1, 3; h 3; & 7 \\
0.201531 & 0.01185477 & h 0; h 1; h 2; s 0; cx 1, 3; h 1; cx 2, 0; & 7 \\
0.095508 & 0.00561810 & h 0; h 1; h 2; s 2; cx 1, 3; cx 2, 1; h 2; & 7 \\
0.084546 & 0.00497331 & h 0; h 1; s 0; cx 0, 2; h 0; cx 1, 3; h 1; & 7 \\
0.092639 & 0.00544933 & h 0; h 1; s 0; cx 0, 2; h 2; cx 1, 3; h 1; & 7 \\
0.175169 & 0.01030405 & h 0; h 2; s 0; cx 0, 1; h 1; cx 2, 3; h 2; & 7 \\
0.304852 & 0.01793247 & h 0; h 1; h 2; h 3; s 0; s 1; s 2; cx 3, 1; & 8 \\
0.324670 & 0.01909825 & h 0; h 1; h 2; h 3; s 0; s 1; s 2; cx 3, 2; & 8 \\
0.197758 & 0.01163285 & h 0; h 1; h 2; h 3; s 0; s 1; cx 2, 0; cx 0, 1; & 8 \\
0.169863 & 0.00999196 & h 0; h 1; h 2; h 3; s 0; s 1; cx 2, 0; cx 2, 1; & 8 \\
0.171991 & 0.01011710 & h 0; h 1; h 2; h 3; s 0; s 1; cx 2, 1; cx 3, 0; & 8 \\
0.016726 & 0.00098386 & h 0; h 1; h 2; h 3; s 0; s 1; cx 3, 1; cx 1, 0; & 8 \\
0.185850 & 0.01093233 & h 0; h 1; h 2; h 3; s 0; s 2; s 3; cx 1, 0; & 8 \\
0.012637 & 0.00074338 & h 0; h 1; h 2; h 3; s 0; s 2; cx 1, 0; cx 3, 2; & 8 \\
0.124522 & 0.00732482 & h 0; h 1; h 2; h 3; s 0; s 2; cx 3, 0; cx 0, 2; & 8 \\
0.136982 & 0.00805779 & h 0; h 1; h 2; h 3; s 0; s 2; cx 3, 2; cx 2, 0; & 8 \\
0.124977 & 0.00735156 & h 0; h 1; h 2; h 3; s 1; s 2; cx 3, 1; cx 1, 2; & 8 \\
0.044544 & 0.00262021 & h 0; h 1; h 2; h 3; s 1; s 2; cx 3, 1; cx 3, 2; & 8 \\
0.148831 & 0.00875475 & h 0; h 1; h 2; s 0; s 1; cx 0, 3; h 0; cx 0, 1; & 8 \\
0.165057 & 0.00970923 & h 0; h 1; h 2; s 0; s 1; cx 0, 3; h 0; cx 1, 3; & 8 \\
0.194340 & 0.01143175 & h 0; h 1; h 2; s 0; s 1; cx 0, 3; h 3; cx 1, 0; & 8 \\
0.223555 & 0.01315028 & h 0; h 1; h 2; s 0; s 1; cx 0, 3; cx 2, 3; h 3; & 8 \\
0.138879 & 0.00816932 & h 0; h 1; h 2; s 0; s 1; cx 1, 3; h 1; cx 1, 0; & 8 \\
0.062420 & 0.00367175 & h 0; h 1; h 2; s 0; s 1; cx 1, 3; h 3; cx 0, 1; & 8 \\
0.054081 & 0.00318126 & h 0; h 1; h 2; s 0; s 1; cx 2, 0; cx 0, 3; h 3; & 8 \\
0.082694 & 0.00486435 & h 0; h 1; h 2; s 0; s 1; cx 2, 1; cx 2, 3; h 2; & 8 \\
0.262084 & 0.01541668 & h 0; h 1; h 2; s 0; s 1; cx 2, 1; cx 2, 3; h 3; & 8 \\
0.009140 & 0.00053763 & h 0; h 1; h 2; s 0; s 2; cx 0, 3; h 0; cx 0, 2; & 8 \\
0.035367 & 0.00208038 & h 0; h 1; h 2; s 0; s 2; cx 0, 3; h 0; cx 2, 3; & 8 \\
0.092308 & 0.00542988 & h 0; h 1; h 2; s 0; s 2; cx 0, 3; h 3; cx 2, 0; & 8 \\
0.338913 & 0.01993607 & h 0; h 1; h 2; s 0; s 2; cx 0, 3; h 3; cx 3, 2; & 8 \\
0.080891 & 0.00475827 & h 0; h 1; h 2; s 0; s 2; cx 0, 3; cx 1, 3; h 3; & 8 \\
0.102957 & 0.00605628 & h 0; h 1; h 2; s 0; s 2; cx 1, 0; cx 0, 3; h 0; & 8 \\
0.229110 & 0.01347703 & h 0; h 1; h 2; s 0; s 2; cx 1, 0; cx 0, 3; h 3; & 8 \\
0.056791 & 0.00334063 & h 0; h 1; h 2; s 0; s 2; cx 1, 0; cx 1, 3; h 3; & 8 \\
0.025700 & 0.00151179 & h 0; h 1; h 2; s 0; s 2; cx 1, 3; h 1; cx 2, 3; & 8 \\
0.262768 & 0.01545696 & h 0; h 1; h 2; s 0; s 2; cx 2, 3; h 3; cx 0, 2; & 8 \\
0.031957 & 0.00187982 & h 0; h 1; h 2; s 0; cx 0, 3; h 3; cx 1, 0; cx 2, 0; & 8 \\
0.155043 & 0.00912018 & h 0; h 1; h 2; s 0; cx 0, 3; cx 1, 0; h 0; cx 2, 3; & 8 \\
0.300150 & 0.01765589 & h 0; h 1; h 2; s 0; cx 0, 3; cx 1, 0; cx 2, 0; h 2; & 8 \\
0.239679 & 0.01409875 & h 0; h 1; h 2; s 0; cx 0, 3; cx 1, 3; h 1; cx 2, 0; & 8 \\
0.051930 & 0.00305471 & h 0; h 1; h 2; s 0; cx 0, 3; cx 1, 3; h 1; cx 2, 3; & 8 \\
0.076754 & 0.00451491 & h 0; h 1; h 2; s 0; cx 0, 3; cx 1, 3; h 3; cx 2, 0; & 8 \\
0.101581 & 0.00597537 & h 0; h 1; h 2; s 0; cx 0, 3; cx 1, 3; h 3; cx 2, 1; & 8 \\
0.132617 & 0.00780101 & h 0; h 1; h 2; s 0; cx 0, 3; cx 1, 3; cx 2, 1; h 2; & 8 \\
0.191823 & 0.01128368 & h 0; h 1; h 2; s 0; cx 1, 0; cx 0, 3; h 0; cx 2, 3; & 8 \\
0.171652 & 0.01009717 & h 0; h 1; h 2; s 0; cx 1, 0; cx 0, 3; cx 2, 3; h 2; & 8 \\
0.030134 & 0.00177258 & h 0; h 1; h 2; s 0; cx 1, 0; cx 0, 3; cx 2, 3; h 3; & 8 \\
0.242716 & 0.01427742 & h 0; h 1; h 2; s 0; cx 1, 0; cx 1, 3; h 1; cx 2, 0; & 8 \\
0.079205 & 0.00465911 & h 0; h 1; h 2; s 0; cx 1, 0; cx 1, 3; cx 2, 1; h 2; & 8 \\
0.183147 & 0.01077338 & h 0; h 1; h 2; s 0; cx 1, 0; cx 2, 0; cx 2, 3; h 3; & 8 \\
0.016263 & 0.00095667 & h 0; h 1; h 2; s 0; cx 1, 3; h 1; cx 2, 0; cx 0, 1; & 8 \\
0.100611 & 0.00591827 & h 0; h 1; h 2; s 0; cx 1, 3; cx 2, 0; cx 0, 1; h 0; & 8 \\
0.126487 & 0.00744043 & h 0; h 1; h 2; s 1; s 2; cx 1, 3; h 1; cx 1, 2; & 8 \\
0.120775 & 0.00710439 & h 0; h 1; h 2; s 1; s 2; cx 1, 3; h 1; cx 2, 3; & 8 \\
0.060643 & 0.00356726 & h 0; h 1; h 2; s 1; s 2; cx 1, 3; h 3; cx 2, 1; & 8 \\
0.180021 & 0.01058947 & h 0; h 1; h 2; s 1; s 2; cx 1, 3; h 3; cx 3, 2; & 8 \\
0.247046 & 0.01453211 & h 0; h 1; h 2; s 1; s 2; cx 2, 3; h 3; cx 1, 2; & 8 \\
0.176296 & 0.01037034 & h 0; h 1; h 2; s 1; cx 0, 3; h 0; cx 1, 0; cx 2, 0; & 8 \\
0.184335 & 0.01084322 & h 0; h 1; h 2; s 1; cx 0, 3; h 0; cx 1, 3; cx 2, 1; & 8 \\
0.024029 & 0.00141344 & h 0; h 1; h 2; s 1; cx 0, 3; h 0; cx 1, 3; cx 2, 3; & 8 \\
0.100179 & 0.00589287 & h 0; h 1; h 2; s 1; cx 0, 3; h 0; cx 2, 0; cx 2, 1; & 8 \\
0.101584 & 0.00597550 & h 0; h 1; h 2; s 1; cx 0, 3; h 0; cx 2, 1; cx 1, 3; & 8 \\
0.041209 & 0.00242404 & h 0; h 1; h 2; s 1; cx 0, 3; h 0; cx 2, 1; cx 2, 3; & 8 \\
0.082703 & 0.00486488 & h 0; h 1; h 2; s 1; cx 0, 3; cx 1, 0; cx 2, 1; h 2; & 8 \\
0.098880 & 0.00581647 & h 0; h 1; h 2; s 1; cx 0, 3; cx 2, 0; cx 2, 1; h 2; & 8 \\
0.099215 & 0.00583615 & h 0; h 1; h 2; s 2; cx 0, 3; h 0; cx 1, 0; cx 2, 0; & 8 \\
0.066052 & 0.00388543 & h 0; h 1; h 2; s 2; cx 0, 3; h 0; cx 1, 3; cx 2, 1; & 8 \\
0.191521 & 0.01126596 & h 0; h 1; h 2; s 2; cx 0, 3; cx 1, 0; h 1; cx 2, 0; & 8 \\
0.086139 & 0.00506700 & h 0; h 1; h 2; s 2; cx 0, 3; cx 1, 0; cx 2, 0; h 2; & 8 \\
0.031068 & 0.00182752 & h 0; h 1; h 2; cx 0, 3; h 0; cx 1, 0; cx 0, 3; cx 2, 0; & 8 \\
0.382225 & 0.02248384 & h 0; h 1; h 2; cx 0, 3; h 0; cx 1, 0; cx 0, 3; cx 2, 1; & 8 \\
0.188316 & 0.01107742 & h 0; h 1; h 2; cx 0, 3; h 0; cx 1, 0; cx 0, 3; cx 2, 3; & 8 \\
0.100905 & 0.00593557 & h 0; h 1; h 2; cx 0, 3; h 0; cx 1, 0; cx 2, 0; cx 2, 3; & 8 \\
0.297486 & 0.01749917 & h 0; h 1; h 2; cx 0, 3; cx 1, 0; h 1; cx 2, 0; cx 2, 1; & 8 \\
0.048851 & 0.00287361 & h 0; h 1; h 3; s 0; s 1; cx 0, 2; h 0; cx 1, 2; & 8 \\
0.159416 & 0.00937744 & h 0; h 1; h 3; s 0; s 1; cx 0, 2; h 2; cx 1, 0; & 8 \\
0.103242 & 0.00607306 & h 0; h 1; h 3; s 0; s 1; cx 0, 2; h 2; cx 2, 1; & 8 \\
0.079588 & 0.00468168 & h 0; h 1; h 3; s 0; s 1; cx 1, 2; h 1; cx 1, 0; & 8 \\
0.241280 & 0.01419294 & h 0; h 1; h 3; s 0; s 1; cx 1, 2; h 2; cx 0, 1; & 8 \\
0.079088 & 0.00465223 & h 0; h 1; h 3; s 0; s 1; cx 3, 0; cx 0, 2; h 0; & 8 \\
0.318606 & 0.01874151 & h 0; h 1; h 3; s 0; s 1; cx 3, 0; cx 0, 2; h 2; & 8 \\
0.178152 & 0.01047950 & h 0; h 1; h 3; s 0; s 3; cx 0, 2; cx 1, 0; h 1; & 8 \\
0.172076 & 0.01012209 & h 0; h 1; h 3; s 0; s 3; cx 0, 2; cx 1, 2; h 1; & 8 \\
0.163125 & 0.00959561 & h 0; h 1; h 3; s 0; s 3; cx 1, 0; cx 1, 2; h 1; & 8 \\
0.054692 & 0.00321719 & h 0; h 1; h 3; s 0; s 3; cx 1, 0; cx 1, 2; h 2; & 8 \\
0.071280 & 0.00419294 & h 0; h 1; h 3; s 0; cx 1, 0; cx 1, 2; h 2; cx 3, 0; & 8 \\
0.035301 & 0.00207651 & h 0; h 1; h 3; s 0; cx 1, 2; h 1; cx 3, 0; cx 0, 1; & 8 \\
0.191611 & 0.01127122 & h 0; h 1; h 3; s 0; cx 1, 2; h 1; cx 3, 0; cx 0, 2; & 8 \\
0.127911 & 0.00752417 & h 0; h 1; h 3; s 0; cx 1, 2; cx 3, 0; cx 0, 1; h 0; & 8 \\
0.044008 & 0.00258869 & h 0; h 1; h 3; s 0; cx 3, 0; cx 0, 2; cx 1, 0; h 0; & 8 \\
0.063755 & 0.00375029 & h 0; h 1; h 3; s 0; cx 3, 0; cx 0, 2; cx 1, 0; h 1; & 8 \\
0.091413 & 0.00537726 & h 0; h 1; s 0; cx 0, 2; h 2; cx 0, 3; cx 1, 0; h 0; & 8 \\
0.105057 & 0.00617985 & h 0; h 1; s 0; cx 0, 2; h 2; cx 1, 0; cx 0, 3; h 0; & 8 \\
0.038709 & 0.00227701 & h 0; h 1; s 0; cx 0, 3; h 0; cx 0, 2; cx 1, 0; h 1; & 8 \\
0.153584 & 0.00903436 & h 0; h 1; s 0; cx 0, 3; h 0; cx 1, 2; h 1; cx 1, 3; & 8 \\
0.134770 & 0.00792765 & h 0; h 1; s 0; cx 0, 3; h 0; cx 1, 2; h 1; cx 2, 3; & 8 \\
0.330392 & 0.01943483 & h 0; h 1; s 0; cx 0, 3; h 3; cx 1, 0; cx 0, 2; h 0; & 8 \\
0.214438 & 0.01261398 & h 0; h 1; s 0; cx 0, 3; h 3; cx 1, 0; cx 0, 2; h 2; & 8 \\
0.134696 & 0.00792327 & h 0; h 1; s 0; cx 0, 3; h 3; cx 1, 0; cx 1, 2; h 2; & 8 \\
0.205824 & 0.01210727 & h 0; h 1; s 0; cx 0, 3; cx 1, 0; h 1; cx 0, 2; h 0; & 8 \\
0.042706 & 0.00251210 & h 0; h 1; s 0; cx 0, 3; cx 1, 0; h 1; cx 0, 2; h 2; & 8 \\
0.090090 & 0.00529938 & h 0; h 1; s 0; cx 1, 0; cx 0, 3; h 0; cx 1, 2; h 2; & 8 \\
0.072129 & 0.00424286 & h 0; h 1; s 0; cx 1, 0; cx 0, 3; h 3; cx 1, 2; h 1; & 8 \\
0.198947 & 0.01170275 & h 0; h 1; s 1; cx 0, 2; h 0; cx 0, 3; cx 1, 3; h 1; & 8 \\
0.108637 & 0.00639042 & h 0; h 1; s 1; cx 0, 2; h 0; cx 0, 3; cx 1, 3; h 3; & 8 \\
0.151433 & 0.00890781 & h 0; h 1; s 1; cx 0, 2; h 0; cx 1, 0; cx 0, 3; h 0; & 8 \\
0.174774 & 0.01028084 & h 0; h 1; s 1; cx 0, 2; h 0; cx 2, 3; cx 1, 2; h 1; & 8 \\
0.147430 & 0.00867238 & h 0; h 1; s 1; cx 0, 2; cx 0, 3; h 3; cx 1, 2; h 1; & 8 \\
0.080042 & 0.00470837 & h 0; h 1; s 1; cx 0, 3; h 0; cx 1, 2; h 2; cx 0, 2; & 8 \\
0.019889 & 0.00116995 & h 0; h 1; s 1; cx 1, 2; h 2; cx 0, 2; cx 2, 3; h 2; & 8 \\
0.075331 & 0.00443125 & h 0; h 1; s 1; cx 1, 2; h 2; cx 0, 2; cx 2, 3; h 3; & 8 \\
0.039551 & 0.00232652 & h 0; h 2; h 3; s 0; s 2; cx 3, 0; cx 0, 1; h 0; & 8 \\
0.246106 & 0.01447680 & h 0; h 2; h 3; s 0; s 3; cx 2, 0; cx 0, 1; h 1; & 8 \\
0.051552 & 0.00303247 & h 0; h 2; h 3; s 0; cx 2, 0; cx 0, 1; h 0; cx 3, 2; & 8 \\
0.139623 & 0.00821313 & h 0; h 2; s 0; cx 0, 3; h 3; cx 2, 0; cx 0, 1; h 1; & 8 \\
0.099269 & 0.00583935 & h 0; h 2; s 0; cx 2, 0; cx 0, 1; h 0; cx 2, 3; h 2; & 8 \\
0.082586 & 0.00485799 & h 0; h 2; s 2; cx 0, 1; h 0; cx 0, 3; cx 2, 3; h 2; & 8 \\
0.242271 & 0.01425123 & h 0; h 2; s 2; cx 0, 1; h 0; cx 1, 3; cx 2, 3; h 2; & 8 \\
\end{longtable}

\newpage



\bibliographystyle{naturemag}

\bibliography{multiqubit_wire_cutting}

\end{document}